\documentclass[aps,pra,twocolumn,10pt]{revtex4-1} 

\usepackage{comment}
\usepackage[T1]{fontenc}
\usepackage{microtype}
\DisableLigatures{encoding = *, family = * }
\usepackage{ifplatform}
\ifwindows\usepackage{lmodern}\fi

\usepackage{amsmath,amsthm,mathtools,amssymb,bm}
\usepackage{physics}

\usepackage[usenames,dvipsnames]{xcolor}
\usepackage{xstring}
\usepackage{titlecaps}
\usepackage[low-sup]{subdepth}

\usepackage{hyperref}
\usepackage{ifpdf}
    \ifpdf
    \else
      \usepackage{hypdvips}
    \fi
\hypersetup{
    breaklinks=true,
    pdfpagemode=UseNone,
    pdfstartview=FitH,
    colorlinks,
    linkcolor={red!50!black},
    citecolor={red!50!black},
    urlcolor={red!30!black}
}
\usepackage{cleveref}

\newenvironment{module}{}{}

\newcommand{\specialcell}[2][t]{\begin{tabular}[#1]{@{}l@{}}#2\end{tabular}}
\usepackage{booktabs}

\RequirePackage{xy}

\xyoption{matrix}
\xyoption{frame}
\xyoption{arrow}
\xyoption{arc}

\entrymodifiers={!C\entrybox}

\newcommand{\qw}[1][-1]{\ar @{-} [0,#1]}
\newcommand{\qwx}[1][-1]{\ar @{-} [#1,0]}

\newcommand{\cw}[1][-1]{\ar @{=} [0,#1]}

\newcommand{\gate}[1]{*+<.6em>{#1} \POS ="i","i"+UR;"i"+UL **\dir{-};"i"+DL **\dir{-};"i"+DR **\dir{-};"i"+UR **\dir{-},"i" \qw}

\newcommand{\multimeasureD}[2]{*+<1em,.9em>{\hphantom{#2}} \POS [0,0]="i",[0,0].[#1,0]="e",!C *{#2},"e"+UR-<.8em,0em>;"e"+UL **\dir{-};"e"+DL **\dir{-};"e"+DR+<-.8em,0em> **\dir{-};{"e"+DR+<0em,.8em>\ellipse^{}};"e"+UR+<0em,-.8em> **\dir{-};{"e"+UR-<.8em,0em>\ellipse^{}},"i" \qw}
\newcommand{\control}{*!<0em,.025em>-=-<.2em>{\bullet}}
\newcommand{\controlo}{*+<.01em>{\xy -<.095em>*\xycircle<.19em>{} \endxy}}
\newcommand{\ctrl}[1]{\control \qwx[#1] \qw}

\newcommand{\targ}{*+<.02em,.02em>{\xy ="i","i"-<.39em,0em>;"i"+<.39em,0em> **\dir{-}, "i"-<0em,.39em>;"i"+<0em,.39em> **\dir{-},"i"*\xycircle<.4em>{} \endxy} \qw}
\newcommand{\multigate}[2]{*+<1em,.9em>{\hphantom{#2}} \POS [0,0]="i",[0,0].[#1,0]="e",!C *{#2},"e"+UR;"e"+UL **\dir{-};"e"+DL **\dir{-};"e"+DR **\dir{-};"e"+UR **\dir{-},"i" \qw}
\newcommand{\ghost}[1]{*+<1em,.9em>{\hphantom{#1}} \qw}
\newcommand{\push}[1]{*{#1}}
\newcommand{\gategroup}[6]{\POS"#1,#2"."#3,#2"."#1,#4"."#3,#4"!C*+<#5>\frm{#6}}

\newcommand{\lstick}[1]{*!R!<.5em,0em>=<0em>{#1}}
\newcommand{\ustick}[1]{*!D!<0em,-.5em>=<0em>{#1}}
\newcommand{\dstick}[1]{*!U!<0em,.5em>=<0em>{#1}}
\newcommand{\Qcircuit}{\xymatrix @*=<0em>}

\newcommand{\pureghost}[1]{*+<1em,.9em>{\hphantom{#1}}}

\newcommand{\phantomtarg}{*+<.02em,.02em>{\xy ="i",
        "i"-<.39em,0em>;"i"+<.39em,0em> **\dir{-}, 
        "i"-<0em,.39em>;"i"+<0em,.39em> **\dir{},
        "i"*\xycircle<.4em>{\dir{}} 
        \endxy} \qw}
\newcommand{\phantomcontrol}{*{\rule{0pt}{.225em}}}

\SilentMatrices

\newcommand*{\qkr}{Quantum Key Repeater}

\newcommand*\locc{{LOCC}}

\newcommand{\keyof}[1]{#1_{\smash k}}
\newcommand{\shieldof}[1]{#1_{\smash s}}
\newcommand{\keyA}{{\keyof{A}}}
\newcommand{\keyB}{{\keyof{B}}}
\newcommand{\key} {{\keyA\keyB}}
\newcommand{\shieldA}{{\shieldof{A}}}
\newcommand{\shieldB}{{\shieldof{B}}}
\newcommand{\shield} {{\shieldA\shieldB}}

\providecommand\bell{}
\renewcommand\bell{{Bell}}
\newcommand*\measurement{\mathcal{M}}
\newcommand\measurements{{\mathbb{M}}}
\newcommand\maps{{\mathbb{L}}}
\newcommand*\corr{{\smash{\hat\Phi}\vphantom\Phi}}

\newcommand*\sep{{SEP}}
\newcommand*\bnot{{BNOT}}
\newcommand*\cnot{CNOT}

\newcommand{\ab} {{A{:}B}}
\newcommand{\atb}{{A{\to}B}}

\newcommand\bgamma{\gamma_{\bell}}
\newcommand{\keyattacked}
{key attacked state}

\newcommand\bit
{state}
\newcommand\eprbit
{maximally entangled \bit}
\newcommand\firstebit
{\emph{{\eprbit}s}}
\newcommand\mebit
{\eprbit}
\newcommand\secretbit
{perfect secret bit}
\newcommand\firstsbit
{\emph{{\secretbit}s}}
\newcommand\privatebit
{private \bit}
\newcommand\firstpbit
{\emph{{\privatebit}s}}
\newcommand\hbit{data hiding \bit}
\newcommand\rent
{repeater rate}
\newcommand\rkey
{key repeater rate}
\newcommand\skey
{key swapping rate}
\newcommand\textrate[1]{rate at which {#1} can be distilled}

\newcommand\repeatersystems{%
    \providecommand{\charlieA}{C}
    \providecommand{\charlieB}{C^{\smash\prime}} 
    \providecommand{\charlie} {{\bm C}}
    \providecommand{\rhoA}{\rho} 
    \providecommand{\rhoB}{\rho^{\smash\prime}} 
    \providecommand{\ac}  {{A{:}\charlieA}}
    \providecommand{\cb}  {{\smash\charlieB{:}B}}
    \providecommand{\accb}{{A{:}\charlieA{:}\smash\charlieB{:}B}}
    \providecommand{\acb} {{A{:}\charlie{:}B}}
    \providecommand{\abc} {{AB{:}\charlie}}
    \providecommand{\cab} {{\charlie{:}AB}}
    \providecommand{\abtc}{{AB\leftrightarrow\charlie}}
    \providecommand{\ctab}{{\charlie\rightarrow AB}}
    \providecommand{\abfc}{{AB\leftarrow\charlie}}
    \providecommand{\sabc}{{AB\mathrlap{{:}\charlie}}}
    \providecommand{\repcutA}{{A{:}\charlie B}}
    \providecommand{\repcutB}{{A\charlie{:}B}}
    \providecommand{\trc}{\tr_{\scriptscriptstyle\charlie}}
    \renewcommand{\achievedrate}{R}
    \renewcommand{\distillmaps}{{\Lambda_\acb}}
    \renewcommand{\distillmap} {{\trc\Lambda}}
}

\newcommand*{\naturals} {{\mathbb{N}}}
\newcommand*{\complex}  {{\mathbb{C}}}
\newcommand*{\integers} {{\mathbb{Z}}}

\usepackage{bbm}
\newcommand*{\one}{\mathbbm{1}}
\newcommand*{\eps}{\varepsilon}
\newcommand*{\id} {{\mathrm{id}}}

\providecommand{\tensor}[1]{^{\otimes #1}}
\renewcommand{\tensor}[1]{^{\otimes #1}}

\newcommand*{\otime}{{\otimes}}
\renewcommand{\lim}{\qopname \relax m{lim\vphantom{p}}}
\renewcommand{\inf}{\qopname \relax m{inf\vphantom{p}}}
\newcommand{\coloneqq}{\mkern-4mu\vcentcolon\mathrel{\mkern-1.2mu}=}

\newcommand*{\proj}[1]{\ketbra{#1}{#1}}

\newcommand*{\leftguil} {{\textnormal\guilsinglleft}}
\newcommand*{\rightguil}{{\textnormal\guilsinglright}}
\newcommand{\onept}{\hspace{-1pt}}
\newcommand{\gstrict}{\leftguil\onept\gamma\hspace{-.5pt}\rightguil}
\newcommand{\strict}[2]{#1^{\onept\leftguil\onept#2\onept\rightguil}}
\newcommand\swapkey{{ \textstyle ^\leftguil\hspace{-3pt}R\mathrlap{\hspace{-.9pt}^\rightguil}^{\,\rightarrow}_D }}
\newcommand{\twisting}{T} 
\newcommand{\kindex}  {j}

\newcommand{\phaseflip}{\mathcal{Z}}
\newcommand{\reversiblemap}{\mathcal{E}}

\newcommand{\distillmap} {\Lambda}
\newcommand{\distillmaps}{{\distillmap_{A{:}B}}}
\newcommand{\achievedrate}{K}
\newcommand{\closeto}[2]{#1 \approx_\eps #2}
\newcommand*{\distillrate}[2]{
    {\lim _{\eps\to0} \lim_{n\to\infty} 
        \sup _{\distillmaps} 
        \left\{
        \achievedrate: 
        \closeto{\distillmap(#1^{\otimes n})}{#2^{n\achievedrate}}
        \right\}
    }
}

\makeatletter
\renewcommand\th@plain     {\thm@notefont{}\itshape   }
\renewcommand\th@definition{\thm@notefont{}\normalfont}
\makeatother

\theoremstyle{plain}
\newcommand{\plainname}{Theorem}
\newtheorem {plain} {\plainname}
\newtheorem*{plain*}{\plainname}
\newtheorem {theorem}[plain]{Theorem}
\newtheorem*{theorem*}       {Theorem}
\newtheorem {lemma}[plain]{Lemma}
\newtheorem*{lemma*}      {Lemma}
\newtheorem {corollary}[plain]{Corollary}
\newtheorem*{corollary*}      {Corollary}
\newtheorem {definition}[plain]{Definition}
\newtheorem*{definition*}      {Definition}

\newtheorem*{conjecture*}      {Conjecture}

\newtheorem*{construction*}      {Construction}

\theoremstyle{definition}
\newtheorem {exampleplain}[plain]{Example}
\newtheorem*{exampleplain*}      {Example}
\newenvironment{example}
{\pushQED{\qed}\renewcommand{\qedsymbol}{$\blacktriangle$}
    \exampleplain}
{\popQED\endexampleplain}    
\newenvironment{example*}
{\pushQED{\qed}\renewcommand{\qedsymbol}{$\blacktriangle$}
    \begin{exampleplain*}}
    {\popQED\end{exampleplain*}}    

\theoremstyle{remark}
\newcommand{\remarkname}{Remark}

\begin{document}


\title{Bell private states}
\title{{\qkr}s: Swapping Perfect Key}
\title{Private States, Quantum Data Hiding and the Swapping of Perfect Secrecy}

\author{Matthias Christandl}
\email{christandl@math.ku.dk}

\author{Roberto Ferrara}
\email{roberto@math.ku.dk}

\affiliation{QMATH, Department of Mathematical Sciences, University of Copenhagen, Universitetsparken 5, 2100 Copenhagen \O, Denmark}

\date{\today}

\begin{abstract}
We derive a formal connection between quantum data hiding and quantum privacy,
confirming the intuition behind the construction of bound entangled states
from which secret bits can be extracted.
We present three main results.
First, we show how to simplify the class of {\privatebit}s
via reversible local operation and one-way communication.
Second, we obtain a bound on the one-way distillable entanglement
of {\privatebit}s in terms of restricted relative entropy measures,
which is tight in many cases and
shows that protocols for one-way distillation of key
from states with low distillable entanglement
lead to the distillation of data hiding states.
Third, we consider the problem of extending the distance of quantum key distribution
with help of intermediate stations.
In analogy to the quantum repeater,
this paradigm has been called the quantum key repeater.
We show that when extending {\privatebit}s with one-way communication,
the resulting rate is bounded by the one-way distillable entanglement.
In order to swap perfect secrecy it is thus essentially optimal to use entanglement swapping.
\end{abstract}

\begin{abstract}
    An important contribution to the understanding of quantum key distribution
    has been the discovery of entangled states from which secret bits,
    but no maximally entangled states,
    can be extracted [Horodecki \emph{et al.}, Phys. Rev. Lett. 94, 200501 (2005)].
    The construction of those states was based on an intuition
    that the quantum mechanical phenomena of data hiding and privacy might be related.
    In this Letter we firmly connect these two phenomena and highlight three aspects of this result.
    First, we simplify the definition of the secret key rate.
    Second, we give a formula for the one-way distillable entanglement
    of certain {\privatebit}s.
    Third, we consider the problem of extending
    the distance of quantum key distribution
    with help of intermediate stations,
    a setting called the quantum key repeater.
    We show that for protocols that first distill {\privatebit}s,
    it is essentially optimal to use the standard quantum repeater protocol
    based on entanglement distillation and entanglement swapping.
\end{abstract}

\maketitle


\section{Introduction}

\begin{module}
Entanglement distillation~\cite{dist_e} is
the process of producing high-fidelity {\firstebit}
from copies of a noisy entangled state $\rho$,
using only Local Operations and Classical Communication (LOCC),
between two parties Alice  and Bob.
The {\eprbit}s can then be used for teleportation, Bell inequality violation, etc.
The \textrate{they} from $\rho$
is called the \emph{distillable entanglement}, $E_D(\rho)$.
Because {\eprbit}s are pure,
they are in product with the environment
 and, therefore,
 measuring them leads to perfectly correlated and perfectly secure pairs of bits, the \emph{{\secretbit}s}. 
It turns out that there exist mixed states, the {\firstpbit},
that also lead to perfectly secure bits just by measurement~\cite{secure}.
While the \emph{distillable key}, $K_D(\rho)$,
is defined as the \textrate{{\secretbit}s}
by local operations and \emph{public} communication,
it was shown that it also equals the \textrate{{\privatebit}s} by LOCC\@.
Proving this equivalence allowed the authors to show
that distillable entanglement and distillable key can be very different~\cite{secure}.
There even exists a low-dimensional experimental realization
of this separation with photonic states~\cite{realpbit}.

In light of this, it is natural to ask
how much the separation extends to general network scenarios,
and in particular whether it persists
if we insert a repeater station between the two parties.
In~\cite{limits} the first examples have been produced of states that,
while having high distillable key,
do not allow for distillation of
significant amounts of the key across 
the repeater station.
This may be an indication that the separation
between the distillable key and distillable entanglement
does not survive in all general network scenarios.

Here we provide a new perspective on key distillation,
and thus quantum key distribution,
by relating {\privatebit}s
to quantum data hiding~\cite{datahiding1,datahiding2}.
This provides a tool for the study
of long-distance quantum key distribution
involving intermediate repeater stations,
where for the first time we are able to show
a close connection with entanglement distillation.
In this framework~\cite{limits},
noisy entanglement is distributed between the end points
and the repeater station
and arbitrary noiseless LOCC protocols are allowed.
If this setting is used to distill {\eprbit}s at the end points then
this is an idealized version of the well-known  quantum repeater
and if it is used  to distill {\privatebit}s it is called a quantum key repeater.
We provide an upper bound on the quantum key repeater rate
with one-way classical communication;
as such, the bound holds also for noisy protocols
that can only lower the rate and thus, if anything, leave room for improvement.
Our results go beyond the use of the partial transpose
and thus apply to states that are not positive under partial transposition (NPT states) as well as states that are invariant under partial transposition (PPT invariant states),
which are out of reach for~\cite{limits}.
\\\indent
The Letter is organized as follows. 
First, we simplify the class of {\privatebit}s,
introducing what we call {\bell} {\privatebit}s.
We show that these states are, for all entanglement-related purposes,
equivalent to {\privatebit}s.
Second, the simplified structure of {\bell} {\privatebit}s
allows us to confirm the intuition that
the separation between the distillable key and distillable entanglement
is due to quantum data hiding.
More precisely, we show that the states with a separation are those
made of a {\eprbit} subject to phase flip error,
where the error information is conserved in data-hiding states.
Such hidden information of the error preserves the key,
but prevents Alice and Bob from correcting the {\eprbit} and distill entanglement.
Third,
as an application to the quantum key repeater with one-way classical communication from the repeater station, we
show that a large class of states and protocols
cannot be used to distill the key across a repeater station
better than by performing entanglement distillation and swapping.

\end{module}

\begin{module}
\newcommand\zkindex{{\kindex}}
\newcommand\zzindex{{0}}


\section{{\privatebit}s}

\begin{module}
Consider two parties Alice and Bob sharing a {\eprbit} $\Phi$
of two qubit systems $\key$, the \emph{key} systems.
Measuring $\Phi$ in the computational basis 
will produce a {\secretbit} with respect to any adversary;
the  postmeasurement state of such a measurement is called a \emph{\keyattacked},
this will play an important role in our results
and will be denoted by a hat ($\hat{\;\;}$): 
\begin{align}
    \label{eq:bell-bits}
    \Phi &\coloneqq {\frac{1}{2}}
          (\ket{00} + \ket{11})
          (\bra{00} + \bra{11})
    \\\nonumber
    \hat\Phi &= {\frac{1}{2}} \left(\proj {00} + \proj{11}\right)\;.
\end{align}
The support of $\hat\Phi$ is known as the \emph{maximally correlated subspace}.
%
Now let  Alice and Bob share additional systems $\shield$,
the \emph{shield} systems.
A \emph{\privatebit} $\gamma$ is a state on $\key\shield$
that generalizes the {\eprbit},
in the sense that measuring $\key$ produces a {\secretbit}
with respect to any adversary. 
$\gamma$ is a {\privatebit} if and only if
it has the form~\cite{secure}:
\begin{align}
    \label{eq:pbit}
    \gamma     &\coloneqq
        \twisting  \left(
        \Phi 
        \otimes \sigma
        \right) \twisting^\dagger
    &
    \hat\gamma &=
        \twisting \left(
        \smash{\hat\Phi}
        \otimes \sigma
        \right) \twisting^\dagger \;.
\end{align}
for some state $\sigma$ on $\shield$ and 
controlled unitary $\twisting$ 
called \emph{twisting};
with no shield systems the only {\privatebit}s are {\eprbit}s.
However, the first example of a {\privatebit}s with low distillable entanglement
was constructed as follows~\cite{secure}:
\begin{align}
    \label{eq:flag-form}
        \gamma &=    p_0  \cdot \phaseflip_\keyB^0(\Phi) \otimes \sigma _0
                  +  p_1 \cdot \phaseflip_\keyB^1(\Phi) \otimes \sigma _1
\end{align}
where $\sigma _\kindex$ are the extremal Werner states~\cite{wernerstates}
and $\phaseflip^\kindex(\varrho)\coloneqq{Z}^\kindex\varrho{Z}^{-\kindex}$ is the $\kindex$th phase flip map,
namely the map that conjugates by the $j$th power of the Pauli $Z$.
The intuition behind the example is the following:
orthogonal data-hiding states $\sigma_\kindex$~\cite{datahiding1,datahiding2}, 
like the Werner states~\cite{wernerhiding},
should hinder the ability to correct the phase flip 
locally and, thus, they should suppress the distillable entanglement,
nevertheless, because the states are orthogonal,
the {\secretbit} is still protected from the environment.
\end{module} %

\begin{module}%
%
%
\providecommand\zkindex{{0\kindex}}%
\providecommand\zzindex{{00}}%
\expandafter\MakeUppercase\privatebit s like the ones in \Cref{eq:flag-form}
are only a special case
(see~\cite{lowdimensional} for different examples);
we call them \emph{{\bell} {\privatebit}s}.
We now show how to convert all {\privatebit}s
into {\bell} {\privatebit}s reversibly
using only LOCC.
We need two generalizations.

We generalize the {\eprbit} to any key systems
of equal finite dimension $|\keyA| = |\keyB|$.
We define the Bell states
$\phi_\zkindex=\proj{\phi_\zkindex}=\mathcal{Z}_{\keyB}^\kindex(\Phi)$
for $j=0,\dots, |\keyB|{-}1$.
Notice that $\{\ket{\phi_\zkindex}\}$
form a basis for the maximally correlated subspace,
which brings us to the next generalization.
We consider any state 
supported only on the maximally correlated subspace of $\key$,
we call such states \emph{key correlated}.
They have no bit-flip error and we can write them as:
\newcommand{\rhoblock}{P_}
\begin{align}
    \rho &= \sum _{\mu\nu} 
        \ketbra{\phi_\mu}{\phi_\nu} 
        \otimes \rhoblock{\mu\nu} 
    \;,
    \label{eq:rho-corr}
\end{align}
where  $\rhoblock{\mu\nu}$ are matrices on $\shield$.
\pagebreak[3]

%
%
\newcommand{\targetA}{\keyA} 
\newcommand{\targetB}{\keyB} 
\newcommand{\target}{{\targetA\targetB}}%
\renewcommand{\bnot}{V}%
To define the reversible LOCC map,
consider two copies of systems $\key$.
Let $\bnot=\cnot_{\keyA\targetA}\!\otime\cnot_{\keyB\targetB}$
be the local unitary illustrated in \Cref{bnot-small},
namely the generalization
of the qubit BNOT~\cite{dist_e}.
It holds that:
    \begin{equation}
    \label{eq:bnot-action}
    \bnot (
    \ket{\phi _\zkindex} _{\target} \!
    \otimes
    \ket{\phi_\mu} _{\key}
    ) =
    \ket{\phi _\zkindex} _{\target} \!
    \otimes Z_\keyB^{-\kindex}
    \ket{\phi _{\mu}} _{\key}.
    \end{equation}

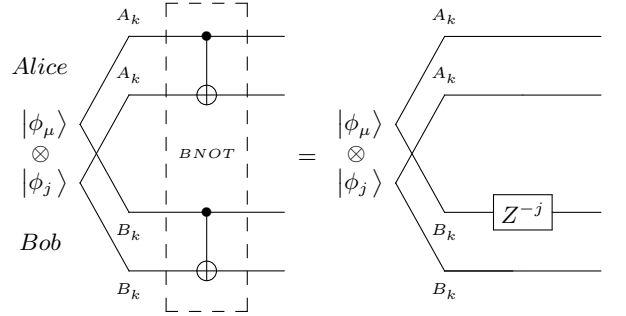
\begin{figure}
\renewcommand{\zkindex}{{\mathrlap\kindex\phantom\mu}}
\renewcommand{\zzindex}{\mu}
\renewcommand{\bnot}{BNOT}
\newcommand{\circuitC}{2em}
\newcommand{\controlone}{\keyA}
\newcommand{\controltwo}{\keyB}
\newcommand{\targetone}{\targetA}
\newcommand{\targettwo}{\targetB}

\begin{module}
\providecommand{\tinymath}[1]{\scalebox{.5}{$#1$}}
\providecommand{\pauli}[2]{X^{#1}Z^{#2}}
\providecommand{\circuitC}{1.4em}
\providecommand{\circuitR}{.5em}
\providecommand{\circuitRfix}{.7em}
\providecommand{\zzindex}{{00}}
\providecommand{\zkindex}{{0k}}

\[
\phantom{\ket{\phi_\zzindex}\;}
\begin{aligned}
\Qcircuit  @C=\circuitC @R=\circuitR @!R=\circuitRfix {
& \ustick{\scriptstyle \controlone} & \ctrl{2} & \qw
\\\lstick{Alice}\\
& \ustick{\scriptstyle \targetone } & \targ    & \qw
\\
\lstick{\ket{\phi _\zzindex}} \ar @{-} [-3,1] \ar @{-} [3,1]
\\
\lstick{\otimes\;\;\,} &&\push{\scriptscriptstyle \bnot}
\\
\lstick{\ket{\phi _\zkindex}} \ar @{-} [-3,1] \ar @{-} [3,1]
\\
& \dstick{\scriptstyle \controltwo} & \ctrl{2} & \qw
\\\lstick{Bob\;}\\
& \dstick{\scriptstyle \targettwo } & \targ    & \qw
\gategroup{1}{3}{9}{3}{2.5em}{--}
}
\end{aligned}
\;=\;
\phantom{\ket{\phi _\zzindex}\;\;}
\begin{aligned}
\Qcircuit  @C=\circuitC @R=\circuitR @!R=\circuitRfix {
& \ustick{\scriptstyle \controlone} & \phantomcontrol\qw & \qw
\\\\
& \ustick{\scriptstyle \targetone } & \qw & \qw
\\
\lstick{\ket{\phi _\zzindex}} \ar @{-} [-3,1] \ar @{-} [3,1]
\\
\lstick{\otimes\;\;\,}
\\
\lstick{\ket{\phi _\zkindex}} \ar @{-} [-3,1] \ar @{-} [3,1]
\\
& \dstick{\scriptstyle \controltwo} & \gate{Z^{-\kindex}} & \qw
\\\\
& \dstick{\scriptstyle \targettwo } & \phantomtarg\qw & \qw
\gategroup{1}{3}{9}{3}{2.5em}{}
}
\end{aligned}
\]

\end{module}
\caption{Quantum circuit for the Bilateral CNOT %
acting on Bell states, 
the core of the map $\reversiblemap$
of \Cref{lemma:reversible-rho}.}
\label{bnot-small}
\end{figure}

%
%
\begin{lemma}%
\label{lemma:reversible-rho}
Define $\reversiblemap: \key \to \target\key$ as
\begin{equation*}
\mathcal{E}(\varrho _{\key}) \coloneqq
\bnot^\dagger \left(
\smash{\hat\Phi} _{\target} 
\otimes \varrho _{\key}
\right) \bnot.
\end{equation*}
Then for any key correlated state $\rho$ (on $\key\shield$):
\begin{align}
    \label{eq:reversible-rho}
    \reversiblemap (\rho) \equiv
    (\reversiblemap _\key {\otimes} \id _\shield) (\rho)
        &=  \frac{1}{|\keyB|}\sum _\kindex 
            \phi _{\kindex} 
            {\otimes} \phaseflip^\kindex_\keyB (\rho)
.
\end{align}
\end{lemma}

\noindent
Because ${\hat\Phi}$ is separable and $\bnot$ is local,
$\reversiblemap$ is one-way LOCC 
(classical communication only from Alice to Bob or vice versa).
$\reversiblemap$ is reversible by inverting $\bnot$ and tracing out the target,
which requires only local operations.
Notice that the output key systems are still $\key$
but the output shield systems are now $\keyA\shieldA\keyB\shieldB $.

\begin{proof}
Using~\eqref{eq:rho-corr},~\eqref{eq:bnot-action} and
$\hat\Phi = \frac{1}{|\keyB|}\sum_\kindex  \phi_\zkindex$
we find
\allowdisplaybreaks%
\begin{align*}
(\mathcal{E}\otimes\id)(\rho)
  &= \sum_{\mu\nu}
     \bnot^\dagger (\hat\Phi \otimes \ketbra{\phi_\mu}{\phi_\nu}) \bnot
     \otimes \rhoblock{\mu\nu}
\\&= \frac{1}{|\keyB|} \sum_{\mu\nu\kindex}
     \bnot^\dagger({\phi_\zkindex} \otimes \ketbra{\phi_\mu}{\phi_\nu})\bnot
     \otimes \rhoblock{\mu\nu}
\\&= \frac{1}{|\keyB|} \sum_{\mu\nu\kindex}
     \phi_\zkindex
     \otimes \phaseflip^\kindex_\keyB(\ketbra{\phi_{\mu}}{\phi_{\nu}})
     \otimes \rhoblock{\mu\nu}
\;.\qedhere
\end{align*}
\end{proof}

%
%
{\bell} {\privatebit}s now come as a special case.
A \emph{{\bell} {\privatebit}} is any {\privatebit} of the form:
\begin{align*}
    \gamma_\bell &= \sum \nolimits_\kindex p_\kindex \cdot \phi _\zkindex \otimes \sigma _\kindex
\end{align*}
where $\sigma _\kindex$ are
arbitrary orthogonal states of  $\shield$
and $p_\kindex$ are arbitrary probabilities.
Notice that $\phaseflip^\kindex_\keyB(\rho)$ are orthogonal, and thus
$\reversiblemap(\rho)$ is a {\bell} {\privatebit},
only when $\rho$ is a {\privatebit}.
Because  $\reversiblemap$ is reversible,
any property of entanglement monotones
(entanglement measures like the distillable entanglement and distillable key)
for {\bell} {\privatebit}s also holds for {\privatebit}s and vice versa.
For example, we can always convert the output of a key distillation protocol
into an approximate {\bell} {\privatebit},
thus simplifying the distillable key
to the \textrate{{\bell} {\privatebit}s}.

\end{module}


\section{Entanglement~Distillation and Quantum~Data~Hiding\label{sec:datahiding}}

\begin{module}%
We now show that
{\bell} {\privatebit}s with low distillable entanglement
are states that hide the phase of the {\eprbit}s from local detection.
Specifically, we give a lower bound
on the one-way distillable entanglement $E_D^\to(\rho)$,
where the communication is one way from Alice to Bob.
This lower bound is the rate achieved
by the best protocol
that starts with a measurement on Alice's shield.

%
%
\newcommand{\rhomeasured}{\tilde\rho}
First for simplicity, let Alice and Bob share a key correlated state
of the form 
\begin{equation}
\label{def:keycorrelated-diagonal}
\rho =  \frac{1}{|\keyB|} \sum_\kindex  
        \phi_\zkindex \otimes \sigma_\kindex
\end{equation}
We now let Alice perform a measurement on her shield and
send the outcome to Bob. 
Then Alice and Bob use the hashing
protocol~\cite{dist_e,dw} and we find:
\begin{align}
\label{eq:datahiding-belldiagonal}
E_D^\to(\rho)
    &\geq \sup _{\measurement\in LO_A}
    \frac{1}{|\keyB|} \sum \nolimits_\kindex
    D(\measurement(\sigma _\kindex) \parallel \measurement(\sigma))
\end{align}
where $\sigma=\frac{1}{|\keyB|}\sum_\kindex\sigma_\kindex$,
$D(\varrho\|\varsigma)=\tr[\varrho\log\varrho-\rho\log\varsigma]$
is the relative entropy, 
and $\measurement$ is a local measurement at Alice
($\measurement\equiv\measurement_{A}\otimes\id_{B}$).
See also the Appendix \Cref{lemma:datahiding-belldiagonal}
for the details.
The relative entropy quantifies the distinguishability between states;
the relative entropy of the measurement outcomes~\cite{piani} 
quantifies how much of this distinguishability is left
when Alice and Bob can only act locally. 
In the particular case of {\privatebit}s,
the $\sigma_\kindex$ states of \Cref{def:keycorrelated-diagonal} are orthogonal;
thus, they are perfectly distinguishable 
and {\emph\kindex} can be recovered with a global measurement.
However, \Cref{eq:datahiding-belldiagonal} implies
that if the distillable entanglement is low,
then the local distinguishability of $\sigma_\kindex$
is low and {\emph\kindex} cannot be determined accurately locally:
the $\sigma_\kindex$ are data hiding~\cite{datahiding1,datahiding2}.

For general key correlated states $\rho$
we can use \Cref{eq:datahiding-belldiagonal}
after using \Cref{lemma:reversible-rho}, this gives 
\begin{align*}
E_D^\to(\rho) 
&\geq \sup  _{\measurement \in LO_A} \frac{1}{|\keyB|} \sum _\kindex
D(\measurement(\phaseflip^\kindex_\keyB (\rho))
\parallel \measurement(\hat\rho))\;.
\label{eq:datahiding-correlated}
\end{align*}
Namely we see 
that because of the reversible map,
we can think of the {\privatebit} itself as a data-hiding state,
where \emph{\kindex} is encoded using the local phase flip.


We can exploit the measurement being local to simplify our bounds.
More precisely, we find that for all local measurements at Alice
    \begin{align*}
D(\measurement (\phaseflip^\kindex_\keyB (\rho))
\|\measurement (\hat\rho))
&=  D(\measurement (\rho) \|\measurement (\hat\rho)) \;.
\end{align*}
Namely, the optimal measurement is independent of the phase flip,
which allows us to remove the phase flip in the formula.
This is an important feature because it suddenly allows us
to regularize~\cite{fekete,regularization} our lower bound.
If $\hat\rho$ is separable, then we can combine the regularized lower bound 
with a known upper bound from~\cite{Emeasures},
and obtain equality with the distillable entanglement
as stated in this theorem 
(see the Appendix \Cref{lemma:datahiding-oneway-full} 
for the details).

\begin{theorem}%
\label{lemma:datahiding-oneway}
    For any key correlated state $\rho$, it holds
    \begin{align*}
        E_D^\to(\rho) 
            &\geq
            D_A(\rho\parallel\hat\rho)
            \coloneqq
            \sup  _{\measurement \in LO_A}
            D(\measurement(\rho)\parallel \measurement(\hat\rho))
        \\
        E_D^\to(\rho) 
            &\geq D_A^{\infty}(\rho\parallel\hat\rho)
            \coloneqq
            \lim _{n\to\infty} \frac{1}{n}
            D_A(\rho^{\otimes n}\parallel\hat\rho^{\otimes n})
    \end{align*}
    If $\hat\rho$ is also separable then:
    \begin{equation*}
    E_D^\to(\rho) =    D_A^\infty(\rho\parallel \hat\rho) \;.
    \end{equation*}
\end{theorem}

\end{module}

\end{module}

\begin{figure}
    \repeatersystems%
    \newcommand{\circuitC}{.81em}
    \newcommand{\target}{\gamma^{n\achievedrate}}
    \renewcommand{\trc}{\tr}
    
    \begin{module}
        \providecommand{\puremultigate}[2]{*+<1em,.9em>{\hphantom{#2}} \POS[0,0]="i",[0,0].[#1,0]="e",!C *{#2},"e"+UR;"e"+UL **\dir{-};"e"+DL **\dir{-};"e"+DR **\dir{-};"e"+UR **\dir{-},"i"} 
        \newcommand{\dw}{\ar@{.} [0,-1]&\ar@{.} [0,-1]} 
        \providecommand{\circuitC}{.7em}
        \providecommand{\circuitR}{.3em}
        \providecommand{\circuitRfix}{.7em}
        \providecommand{\target}{\Phi^{ne}}

        \[
        \phantom{{\rhoB}^{\otimes n}}\hspace{.5em}
        \begin{aligned}
        \Qcircuit  @C=\circuitC @R=\circuitR @!R=\circuitRfix {%
            &\ustick{\scriptstyle {Alice}\;\;}
            &\qw& \multigate{+1}\empty &&&&\qw[-4]&\dw& \multigate{+1}\empty &\qw&\qw
            \\
            &&&\pureghost\empty&\controlo\cw&&&&&\controlo&\pureghost\empty\cw
            \\\lstick{{\rhoA}^{\otimes n}}\ar @{-} [-2,1] \ar @{-} [2,1]&&&&&&&&&\\
            &&&&&\controlo\ar @{=}[-2,-1]&\puremultigate{4}\empty\cw&\controlo\cw&
            \\
            &&&&&\qw[-4]&\ghost\empty&\qw&\dw&\qw&\multimeasureD{2}{\scriptstyle \trc}
            \\&{\scriptstyle Charlie\quad} \\ 
            &&&&&\qw[-4]&\ghost\empty&\qw&\dw&\qw&\ghost{\scriptstyle \trc}
            \\
            &&&&&\controlo\ar @{=}[2,-1]&\pureghost        \empty\cw&\controlo\cw&
            \\\lstick{{\rhoB}^{\otimes n}} \ar @{-} [-2,1] \ar @{-} [2,1]&&&&&&&&&\\
            &&&\pureghost\empty&\controlo\cw&&&&&\controlo&\pureghost\empty\cw
            \\
            &\dstick{\scriptstyle Bob \;}
            &\qw& \multigate{-1}\empty &&&&\qw[-4]&\dw& \multigate{-1}\empty &\qw&\qw
            \gategroup{1}{4}{11}{11}{.5em}{--}
            \save"2,9"."3,8"  \ar@{:}"4,8"  \restore 
            \save"4,9"."3,10" \ar@{:}"2,10" \restore 
            \save"10,9"."9,8" \ar@{:}"8,8"  \restore 
            \save"8,9"."9,10" \ar@{:}"10,10"\restore 
        }
        \end{aligned}
        \;\;\approx _\eps \;
        \target\hspace{.25em}
        \begin{aligned}
        \Qcircuit  @C=\circuitC @R=\circuitR @!R=\circuitRfix {
            &\ustick{\scriptstyle Alice}&\qw
            \\\\\\\\\\
            \ar @{-} [-5,1] \ar @{-} [5,1]
            \\\\\\\\\\ \push{\rule{0pt}{1.08em}}
            &\dstick{\scriptstyle Bob}  &\qw
        }
        \end{aligned}
        \]
        
    \end{module}
    
    \caption{
	\label{figure:repeater}
        Quantum circuit for the key repeater protocols
        in a single node repeater.
        The dashed box is a tripartite LOCC protocol.
        The double lines are the classical communication.}%
\end{figure}
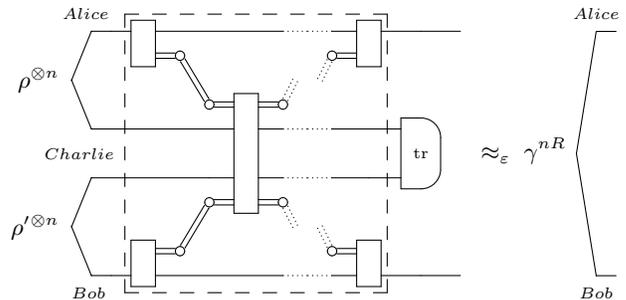


\section{Quantum Key Repeaters}
%
%
\begin{module}
\repeatersystems%
\renewcommand{\charlie}{{\charlieA\charlieB}}%
We now apply our findings to long-distance quantum communication,
where noise prevents Alice and Bob from sharing entanglement and thus secrecy,
and where an intermediate repeater station, Charlie,
is necessary to mediate the entanglement.

More precisely, let Alice and Charlie ($A$ and $C$) share $\rhoA$
and  Charlie and Bob ($\charlieB$ and $B$) share $\rhoB$.
While the goal of a quantum repeater
 is to distill {\eprbit}s between Alice and Bob~\cite{repeater},
the goal of a quantum key repeater is to distill {\secretbit}s 
or, equivalently, {\privatebit}s~\cite{limits}, see \Cref{figure:repeater}.
The best rate for this task
is called the quantum \emph{\rkey}, $R_D(\rhoA,\rhoB)$
.
Realistic repeaters have multiple stations;
however, we reduce to a single station by grouping them into one, 
which can only increase the rate.
The reduction to a single station thus provides upper bounds
without loss of generality.

$\rhoA$ and $\rhoB$
are usually generated by sharing {\eprbit}s 
through noisy channels (Choi-Jamio{\l}kowski states). 
While clever channel codes may
reach higher rates~\cite{concatenated,repeater-generations},
note that the free classical side information 
allows us in most realistic channels 
to implement the codes via teleportation 
from the Choi-Jamio{\l}kowski state
(e.g.\ depolarizing channel)~\cite{dist_e}. 
Thus our upper bounds also apply to such codes and channels 
(see also~\cite{alex-capacities}).

The optimal noise-free protocol for the quantum repeater
performs entanglement distillation
between Alice and Charlie, and between Charlie and Bob,
followed by entanglement swapping.
This results in the rate
$\min\{E_D(\rhoA),E_D(\rhoB)\}$,
but in the quantum key repeater setting
the situation is less clear.
In alternative to the mentioned protocol,
Alice and Charlie can distill {\privatebit}s,
and use the {\eprbit}s distilled by Charlie and Bob to teleport Charlie's part
of the {\privatebit}s.
If $E_D(\rhoB)$ is larger than
the {\privatebit}s size at Charlie's,
then the rate of this ``trivial'' protocol equals $K_D(\rhoA)$
and thus it will be positive even 
when $\rhoA$ has zero distillable entanglement~\cite{secure,lowdimensional}. 
In short, while for quantum repeaters
the active area of research studies the effect of noisy operations,
for quantum key repeaters there are open questions
even with perfect operations.

We will consider the \emph{one-way {\rkey}} variation,
also introduced in~\cite{limits}.
In this variation,
Alice and Bob's communication with the repeater station Charlie is only one way:
Charlie can send messages to Alice and Bob but not vice versa.
Alice and Bob \emph{can} still communicate normally with each other.
We denote this rate with $R_D^{\charlie\to\ab}$, or simply  $R_D^\rightarrow$.
In~\cite{limits} the question was posed 
whether there exist non-trivial protocols
beyond distillation and swapping,
but only negative examples were found.
Here we show that for a large class of states and protocols,
the one-way distillable entanglement  is an upper bound on the one-way {\rkey}, and thus distillation and swapping are essentially optimal
and far from the trivial upper bounds $K_D(\rhoA)$ and $K_D(\rhoB)$.

We need a general upper bound which follows from~\cite[Theorem 4]{limits}:
\begin{equation}
\label{lemma:repeatersep-oneway}
R_D^{\charlie\to\ab}(\rhoA_{A\charlieA},\rhoB_{\charlieB B})
\leq D_{\charlie}^\infty(\rhoA_{A\charlieA}\otimes\rhoB_{\charlieB B} 
\| \sigma)
\end{equation}
for any state $\sigma$ separable in the $A\charlie{:}B$ or $A{:}\charlie B$ cut.
So far, this bound could only be estimated 
via a relaxation that only works for states that are PPT (Positive under Partial Transposition).
Choosing $\sigma=\hat\rhoA\otimes\hat\rhoB$
and applying \Cref{lemma:datahiding-oneway}
to \Cref{lemma:repeatersep-oneway}
now shows the following corollary, independently of the partial transpose.
\begin{corollary}%
\label{lemma:rded-oneway}
For any key correlated states $\rhoA$ and $\rhoB$
with at least one separable {\keyattacked}, it holds:
\[ R_D^{\charlie\to\ab}(\rhoA_{A\charlieA},\rhoB_{\charlieB B})
\leq E_D^{\charlieA\to A} (\rhoA_{AC}\otimes\rhoB_{CA})\;.\]
\end{corollary}


Since all {\privatebit}s are NPT
(Non-positive under Partial Transposition)~\cite{paradigms},
this gives the first examples of NPT states
with a high distillable key but low one-way {\rkey}. 

\begin{example*}%
\label{example:main}
    Consider the following {\bell} {\privatebit}
    (see~\cite{secure} and \Cref{eq:flag-form}):
    \[\gamma
    = \frac{1}{2} \left( 1 + \frac{1}{d} \right) 
    \phi_0 \otimes \sigma_0 
    + \frac{1}{2} \left( 1 - \frac{1}{d} \right) 
    \phi_1 \otimes \sigma_1 
    \]
    where $\sigma_0$ and $\sigma_1$ 
    are, respectively, the symmetric and anti-symmetric states
    in $\complex^{d}\otimes\complex^{d}$ 
    (the extreme Werner states~\cite{wernerstates}
    which are known to be data-hiding states~\cite{wernerhiding}).
    Since distillable entanglement is upper bounded
    by the log-negativity $E_N$~\cite{logneg},
    we have the following upper bound which vanishes for large $d$:
    \begin{equation}
    \label{eq:repeater-main-example}
    R_D^\to(\gamma,\gamma)
    \leq 2 E_D^\to(\gamma)
    \leq 2 E_N(\gamma)
    = 2 \log \left(1+\frac{1}{d} \right)
    \;.
    \end{equation}
    This state was implemented experimentally for $d=2$~\cite{realpbit}.
    The key was distilled at a rate
    $K\approx{0.69}$,
    enough to break the bound at
    $E_N(\gamma)=\log\frac{3}{2}\approx 0.58$.
    However, because of the factor of 2 in \Cref{eq:repeater-main-example},
    an implementation with $d=4$ at the same key rate is required
    for the same proof of concept.
    Still, scaling up the implementation
    should be experimentally feasible,
    since in $d=4$ the gate used (swap)
    is tensor product of qubit gates.
    In the Appendix \Cref{example:ppt-invariant}
    we show how to apply \Cref{lemma:rded-oneway} to some PPT invariant states.
\end{example*}

\end{module}


\section{Conclusions}

\begin{module}%
\repeatersystems%
\Cref{lemma:rded-oneway} bounds the {\rkey}
of a restricted class of states,
but it also generalizes to all states if we restrict the protocols to
first distill {\privatebit}s with separable {\keyattacked} between the nodes
and then try to repeat.
In the Appendix \Cref{def:key-swapping}
we define a new {\rkey} $\swapkey(\rhoA,\rhoB)$
from these protocols and prove that
for all states, this rate is upper bounded by $E_D^\to (\rhoA \otimes \rhoB)$.
The restricted protocols still include 
one-way entanglement distillation and swapping;
thus, the new {\rkey} is still lower bounded by the minimum
of the one-way distillable entanglements.
While being restrictive,
we would like to stress that the communication between Alice and Bob 
is two way, and also that if the two-way step is limited to bipartite distillation between the nodes, we can always apply the result to the outcomes of the distillation.
In particular even if the two-way recurrence protocol is used to distill between the nodes, as in the case of heralded entanglement generation and purification,
we can apply the bound on $\swapkey$ to the outputs of the recurrence protocol.
The bound also applies to key repeater schemes based on quantum error correction.
The link with outgoing communication from the station is trivially covered.
For the link with incoming communication, the bound on $\swapkey$ applies to the output of the code (as mentioned above), since usually the code is decoded or corrected at the station rendering it a bipartite distillation protocol.
As such, we can apply our bound in some way to most repeater schemes
(see also~\cite{repeater-generations} and references therein 
for an overview)
and where it applies, any attempt to improve the rate
of key distillation above that of entanglement distillation will not work.
For example, attempting to use the noisy processing protocol\cite{renner2005information} would yield no advantage.
We are not aware that there exist any protocol that contains a truly two-way tripartite step.


Finally, we note that, 
because optimal one-way protocols exist when close to the target states,
the optimal two-way protocols are composed of a two-way ``lift-off'' protocol followed by a one-way ``conclusion'' protocol~\cite{kretschmann2004tema}.

We leave as an open problem whether \Cref{lemma:rded-oneway}
generalizes to all states and protocols, including two-way communication.
Such a result would show that all
entangled states with zero distillable entanglement, 
including those with distillable key, have zero {\rkey}.
Another open problem, called the PPT\textsuperscript2 conjecture~\cite{ppt2},
asks whether swapping PPT states in all dimensions always yields separable states.
If the conjecture is true,
then it would imply that all PPT states have zero {\rkey}.
In that, the results here presented support the conjecture.
Since our results are asymptotic in nature,
they give a complementary view on the PPT\textsuperscript2 conjecture
to that of the study of swapping specific states in specific dimensions.

The connection made between key distillation, entanglement distillation
and quantum data hiding raises the possibility of finding
a \textrate{{\hbit}s}, $H_D$
(which we refrain from defining formally).
Namely, in performing entanglement distillation on {\privatebit}s,
it may be possible to retain the undistillable correlations
into {\hbit}s with zero distillable entanglement
so that they could be used as a resource, such that
\[K_D(\rho) = H_D(\rho) + E_D(\rho).\]

\end{module}

\hypersetup{bookmarksdepth=-2}

\begin{acknowledgments} 
We thank Alexander M\"uller-Hermes, C\'ecilia Lancien, and M\=aris Ozols for helpful discussions.
We acknowledge financial support from
the European Research Council (ERC Grant Agreement No. 337603),
the Danish Council for Independent Research (Sapere Aude),
and VILLUM FONDEN via the QMATH Centre of Excellence (Grant No. 10059).
\end{acknowledgments}



\section*{Appendices}
In these appendices we present the background concepts
used in this article, we show how to further apply our findings
to more complex repeater scenarios and we provide more examples.

We begin with a review of the generalized {\bell} states
and of their properties, especially with respect to the
bilateral CNOT which is fundamental to the reversible map. 
Then we analyse {\bell} {\privatebit}s further and present some minor properties.

After this, we move onto entanglement measures.
First we review the various entanglement measures based
on the relative entropy and its restriction to quantum measurements;
this heavily relies on the work made in~\cite{piani}.
We then explain all the distillation
rates mentioned in this paper and give their explicit definition, 
this includes the various forms of
distillable entanglement, distillable key and {\rkey}.
We also discuss in details some known upper bounds on these distillation rates
that were used in the main text.

With these concepts in place we present further applications of our results:
the one-way key swapper (a novel repeater rate) and the single-copy repeater rate.
Finally, we give further examples of states with vanishing one-way {\rkey},
which include NPT states, PPT states and PPT invariant states.

Last but not least, the reader can find an exhaustive list of notations in
\Cref{table:notation} and \Cref{table:notation-repeater}.

\appendix

\section{{Generalized Bell states}}

\begin{module}
    Consider a two qubit system ($\complex^2\otimes\complex^2$)
    with the {\eprbit}:
    \[\ket\Phi \coloneqq \frac{1}{\sqrt 2} (\ket{00} + \ket{11})\;.\]
    We can write the Bell states as
    bit flips and phase flips of the {\eprbit}
    where the bit flip and phase flip unitaries are 
    \begin{align*}
    X &\coloneqq \quantity( \begin{matrix} 0&1\\1&0\end{matrix})
    &
    Z &\coloneqq \quantity( \begin{matrix} 1&0\\0&-1\end{matrix}) \;,
    \end{align*}
    and they generate the Bell states
    when acting on a single qubit of a {\eprbit}.
    In general, acting on one qubit or the other yields different bases.
    We will choose to act always on the first qubit and so we define
    \[\ket{\phi_{ij}} \coloneqq (X^i Z^j \otimes \one) \ket{\Phi}\qquad i,j=0,1\]
    which gives the Bell states:
    \begin{align*}
    \ket{\phi_{00}} &= \frac{1}{\sqrt 2} (\ket{00} + \ket{11}) = \ket\Phi
    \\
    \ket{\phi_{01}} &= \frac{1}{\sqrt 2} (\ket{00} - \ket{11})
    \\
    \ket{\phi_{10}} &= \frac{1}{\sqrt 2} (\ket{10} + \ket{01})
    \\
    \ket{\phi_{11}} &= \frac{1}{\sqrt 2} (\ket{10} - \ket{01})\;.
    \end{align*}
    
    The generalized Bell states are defined in a similar way.
    We consider now a two qudit system ($\complex^d\otimes\complex^d$)
    whose {\eprbit} is now:
    \begin{align*}
    \ket{\Phi} &= \frac{1}{d} \sum _{i\in\integers_d} \ket{ii}\;,
    \end{align*}
    {where $\integers_d = \integers/d \integers\equiv\{0,\dots,d-1\}$
        is the cyclic additive group of order $d$ on the integers,
        namely the integers with addition modulo $d$: $i+j= i+j \mod{d}$.
        The unitary generalization of the bit and phase flip are:}
    \begin{align*}
    X &\coloneqq \sum _{j\in\integers_d} \ketbra{j+1}{j}\;,
    &
    Z &\coloneqq \sum _{j\in\integers_d}  \omega^j \proj{j}\;,
    \end{align*}
    where $\omega= e^{i \frac{2\pi}{d}}$
    is the $d$'th root of unity.
    Just as before, the Bell states are now defined
    using powers of $X$ and $Z$.
    \begin{definition}[Generalized Bell states]%
        \label{def:bellbasis}
        \begin{align*}
        \ket{\phi _{ij}}&\coloneqq (X^i Z^j \otime \one) \ket{\Phi}
        \end{align*}
        for $i,j=\integers_d$.
    \end{definition}
    Finally, we further define the following notation for the density matrix of these states:
    \begin{align*}
    \Phi &\coloneqq  \proj{\Phi}\;,
    &
    \phi _{ij} &= \proj{\phi _{ij}} \;.
    \end{align*}
    As mentioned in the main text, we denote with the hat ($\hat{\;\;}$)
    the operation that measures the key systems in the computational basis.
    For the case of the {\eprbit} this measurement yields 
    the \emph{maximally correlated state}:
    \begin{equation*}
    \hat\Phi = \frac{1}{d} \sum _{i\in\integers_d} \proj{ii} \;.
    \end{equation*}
    The maximally correlated state is a uniform mixture of orthogonal pure states,
    we thus define the \emph{maximally correlated subspace}
    as the support of the maximally correlated state.
    We call $\one_\corr\coloneqq d \hat\Phi$ the projector onto this subspace.
    Measuring any of $\phi_{0j}$ (the Bell states with phase flips only)
    in the computational basis yields the same maximally correlated state
    \[\hat\phi_{0j} = \hat\Phi \quad\forall\ j \]
    and $\one_{\hat\Phi}$ can be written as a mixture of Bell states
    in the following way:
    \begin{align*}
    \one_\corr
    &
    = \sum _{i\in\integers_d} \proj{ii}
    = \sum _{i,j\in\integers_d} \delta_{ij} \ketbra{ii}{jj}
    \\&
    = \sum _{{i,}j,k\in\integers_d}
    \frac{1}{d} \omega ^{k(i-j)} \ketbra{ii}{jj}
    = \sum _{k\in\integers_d} \proj{\phi _{0k}}\;.
    \\&
    = \sum _{k\in\integers_d} \phi _{0k}\;.
    \end{align*}
    For the remainder of the article, we will omit $\integers_d$
    and, unless otherwise stated,
    all indexes will be summed over the cyclic group,
    where the  order is given by the dimension of the corresponding Hilbert space.
\end{module}

%
%
\begin{module}
    \newcommand{\scdot}{\;\cdot\;}
    
    \begin{figure}
        \providecommand{\circuitC}{1.5em}
        \providecommand{\circuitRfix}{1em}
        \setlength{\baselineskip}{0pt}
        
        \begin{module}
            \providecommand{\tinymath}[1]{\scalebox{.5}{$#1$}}
            \providecommand{\pauli}[2]{X^{#1}Z^{#2}}
            \providecommand{\circuitC}{1em}
            \providecommand{\circuitR}{.5em}
            \providecommand{\circuitRfix}{.7em}
            \providecommand{\controlone}{C_1}
            \providecommand{\controltwo}{C_2}
            \providecommand{\targetone}{T_1}
            \providecommand{\targettwo}{T_2}
            \[
            \phantom{\ket\Phi\;}
            \begin{gathered} \Qcircuit  @C=\circuitC @R=\circuitR @!R=\circuitRfix {
                & & \ustick{\scriptstyle \controlone}\qw & \gate{\scriptstyle \pauli ij} & \ctrl{2} & \qw
                \\\\
                & & \ustick{\scriptstyle \targetone }\qw & \gate{\scriptstyle \pauli kl} & \targ    & \qw
                \\
                \lstick{\ket\Phi} \ar @{-} [-3,1] \ar @{-} [3,1]
                \\
                \lstick{\otimes\;\,} &&&&{\scriptscriptstyle \bnot}
                \\
                \lstick{\ket\Phi} \ar @{-} [-3,1] \ar @{-} [3,1]
                \\
                & & \ustick{\scriptstyle \controltwo}\qw & \qw & \ctrl{2} & \qw
                \\\\
                & & \ustick{\scriptstyle \targettwo }\qw & \qw & \targ    & \qw
                \gategroup{1}{5}{9}{5}{1.9em}{--}
            }\end{gathered}
            \;=\;
            \phantom{\ket\Phi\;}
            \begin{gathered}\Qcircuit  @C=\circuitC @R=\circuitR @!R=\circuitRfix {
                & & \ustick{\scriptstyle \controlone}\qw & \gate{\scriptstyle \pauli {i  }{j\tinymath-l}}     & \qw
                \\\\
                & & \ustick{\scriptstyle \targetone }\qw & \gate{\scriptstyle \pauli {k\tinymath+i}{l  }} \qw & \qw
                \\
                \lstick{\ket\Phi} \ar @{-} [-3,1] \ar @{-} [3,1]
                \\
                \lstick{\otimes\;\,}
                \\
                \lstick{\ket\Phi} \ar @{-} [-3,1] \ar @{-} [3,1]
                \\
                & & \ustick{\scriptstyle \controltwo}\qw & \qw & \qw
                \\\\
                & & \ustick{\scriptstyle \targettwo }\qw & \push{\rule{0pt}{1.85em}}\qw & \qw
            }\end{gathered}
            \]
            
        \end{module}
        \caption{\label{figure:bnot}%
            Effect of the BNOT on Bell states, see \Cref{lemma:bnot}.}
        \renewcommand{\circuitC}{.7em}
        \begin{minipage}[t][][t]{.567\columnwidth}
            {\newlength\bnotheight%
                \settoheight\bnotheight{\vbox{ %
                        %
                        \begin{module}
                            \providecommand{\tinymath}[1]{\scalebox{.5}{$#1$}}
                            \providecommand{\pauli}[2]{X^{#1}Z^{#2}}
                            \providecommand{\circuitC}{1em}
                            \providecommand{\circuitR}{1em}
                            \providecommand{\circuitRfix}{1em}
                            \[
                            \phantom{\ket\Phi\hspace{3.7pt}}
                            \begin{gathered} \Qcircuit  @C=\circuitC @R=\circuitR @!R=\circuitRfix {
                                & & \ctrl{1} & \qw\\
                                \lstick{\scriptstyle \ket\Phi}\ar @{-} [-1,1] \ar @{-} [1,1]& & \targ    &\qw\\
                                \lstick{\scriptstyle \ket\Phi}\ar @{-} [-1,1] \ar @{-} [1,1]& & \ctrl{1} &\qw\\
                                & & \targ &\qw
                            } \end{gathered}
                            =
                            \phantom{\ket\Phi}
                            \begin{gathered} \Qcircuit  @C=\circuitC @R=\circuitR @!R=\circuitRfix {
                                & &\qw&\qw\\
                                \lstick{\scriptstyle \ket\Phi}\ar @{-} [-1,1] \ar @{-} [1,1]& & \qw &\qw\\
                                \lstick{\scriptstyle \ket\Phi}\ar @{-} [-1,1] \ar @{-} [1,1]& & \qw &\qw\\
                                & &\push{\rule{0pt}{.85em}}\qw&\qw
                            } \end{gathered}
                            \]
                        \end{module}
                }}\parbox[t][\bnotheight][c]{\linewidth}{%
                    \begin{module}
                        \providecommand{\tinymath}[1]{\scalebox{.5}{$#1$}}
                        \providecommand{\pauli}[2]{X^{#1}Z^{#2}}
                        \providecommand{\circuitC}{1em}
                        \providecommand{\circuitR}{1em}
                        \[
                        \begin{gathered} \Qcircuit  @C=\circuitC @R=\circuitR  {
                            & \gate{\scriptstyle \pauli ij}& \ctrl{1} & \qw\\
                            & \gate{\scriptstyle \pauli kl}&\targ \qw & \qw\\
                        }\end{gathered}
                        =
                        \begin{gathered}
                        \Qcircuit  @C=\circuitC @R=\circuitR {
                            & \ctrl{1}& \gate{\scriptstyle \pauli {k}{j\tinymath-l}} & \qw\\
                            & \targ   & \gate{\scriptstyle \pauli {k\tinymath+i}{l}} & \qw\\
                        }\end{gathered}
                        \]
                    \end{module}
            }}\caption{Effect of commuting CNOT with arbitrary bit/phase flips.}%
            \label{figure:clifford}
        \end{minipage}
        \hfill
        \begin{minipage}[t][][t]{.4\columnwidth}
            \begin{module}
                \providecommand{\tinymath}[1]{\scalebox{.5}{$#1$}}
                \providecommand{\pauli}[2]{X^{#1}Z^{#2}}
                \providecommand{\circuitC}{1em}
                \providecommand{\circuitR}{1em}
                \providecommand{\circuitRfix}{1em}
                \[
                \phantom{\ket\Phi\hspace{3.7pt}}
                \begin{gathered} \Qcircuit  @C=\circuitC @R=\circuitR @!R=\circuitRfix {
                    & & \ctrl{1} & \qw\\
                    \lstick{\scriptstyle \ket\Phi}\ar @{-} [-1,1] \ar @{-} [1,1]& & \targ \qw &\qw\\
                    \lstick{\scriptstyle \ket\Phi}\ar @{-} [-1,1] \ar @{-} [1,1]& & \ctrl{1}\qw &\qw\\
                    & & \targ &\qw
                } \end{gathered}
                =
                \phantom{\ket\Phi}
                \begin{gathered} \Qcircuit  @C=\circuitC @R=\circuitR @!R=\circuitRfix {
                    & &\qw&\qw\\
                    \lstick{\scriptstyle \ket\Phi}\ar @{-} [-1,1] \ar @{-} [1,1]&  &\qw&\qw\\
                    \lstick{\scriptstyle \ket\Phi}\ar @{-} [-1,1] \ar @{-} [1,1]&  &\qw&\qw\\
                    & &\push{\rule{0pt}{.85em}}\qw&\qw
                } \end{gathered}
                \]
            \end{module}
            \caption{Bell states with no bit/phase flips are invariant under BNOT.}%
            \label{figure:bnot-invariant}
        \end{minipage}
        
    \end{figure}
    
    \bigskip
    The bilateral $CNOT$
    is a gate of four systems obtained by applying two {\cnot} gates,
    two of the systems will be the controls
    and the others will be the targets.
    In~\cite{dist_e} it was called Bilateral XOR (BXOR)
    and it was defined only for qubit systems.
    Here we use the generalized CNOT on
    $\complex^d \otimes \complex^d$:
    \[ \cnot \coloneqq \sum_{ab} \ketbra{a, a{+}b}{ab}.\] 
    The straightforward generalization of the BXOR gate,
    the \emph{bilateral {\cnot}}, is then:
    \[\bnot _{CT} \coloneqq \cnot_{C_1T_1}\otimes \cnot_{C_2T_2}\]
    where systems $C=C_1C_2$ are the control qudits,
    systems $T=T_1T_2$ are the target qudits
    and all qudits have the same size, i.e.\ $|C_1|=|C_2|=|T_1|=|T_2|=d$.
    Like in~\cite{dist_e}, our interest in the gate lies in its effect on Bell states.
    Notice how the $\bnot$ is a local operation
    as long as the system is partitioned as $C_1T_1 {:} C _2T_2$.
    
    \begin{lemma}%
        \label{lemma:bnot}
        For all $i,j,k,l\in\integers_d$:
        \[\bnot_{CT} \cdot \ket{\phi _{ij}}_C\otimes \ket{\phi _{kl}}_T
        = \ket{\phi _{i,j-l}}_C \otimes \ket{\phi _{k+i,l}}_T\]
    \end{lemma}
    \begin{proof}
        \begin{align*}
        &\bnot  \cdot \ket{\phi _{ij}}\otimes \ket{\phi _{kl}}
        \\&
        = \bnot \cdot \frac{1}{d} \sum _{ab}
        \omega^{ja} \ket{a{+}i,a} \otimes \omega^{lb}\ket{b{+}k,b}
        \\&
        = \frac{1}{d} \sum _{ab}
        \omega^{ja} \ket{a{+}i,a} \otimes \omega^{lb}\ket{b{+}k{+}a{+}i,b{+}a}\;.
        \intertext{Now we make a change of variable $\tilde b = b + a$:}
        &
        = \frac{1}{d} \sum _{a \tilde b}
        \omega^{ja} \ket{a{+}i,a} \otimes \omega^{l(\tilde b - a )}\ket{\tilde b{+}k{+}i,\tilde b}
        \\&
        = \frac{1}{d} \sum _{a \tilde b}
        \omega^{(j-l)a} \ket{a{+}i,a} \otimes \omega^{l\tilde b}\ket{\tilde b{+}k{+}i,\tilde b}
        \\&
        = \ket{\phi _{i,j-l}}\otimes \ket{\phi _{k+i,l}}
        \;. \qedhere
        \end{align*}
    \end{proof}
    
    An alternative and maybe more intuitive way to prove \Cref{lemma:bnot},
    is to notice that $\Phi \otimes \Phi$
    is  invariant under the action of the $\bnot$, 
    see \Cref{figure:bnot-invariant},
    and that $\cnot$ is a Clifford gate with a simple update rule,
    see \Cref{figure:clifford}.
    Namely, it holds that
    \begin{equation}
    \label{eq:bnot-invariant}
    \bnot_{CT} \cdot
    (\ket\Phi _C \otimes \ket\Phi  _T) =
    \ket\Phi _C \otimes \ket\Phi _T
    \end{equation}
    and
    \begin{equation}
    \label{eq:cnot-clifford}
    \cnot \cdot X^i Z^j \otime X^k Z^l = X^i Z^{j-l} \otime X^{i+k} Z^l \cdot \cnot^\dagger.
    \end{equation}
    Applying \Cref{eq:bnot-invariant} and \Cref{eq:cnot-clifford}
    to \Cref{def:bellbasis} proves \Cref{lemma:bnot}
    as displayed in \Cref{figure:bnot}.
    From \Cref{lemma:bnot}, it follows in particular that
    \begin{equation*}
    \label{eq:bnot-phi}
    \bnot^\dagger \cdot \phi _{00}\otimes \phi _{0\kindex} \cdot \bnot = \phi _{0\kindex}\otimes \phi _{0\kindex}
    \end{equation*}
    which is what we use in \Cref{lemma:reversible-rho} (see main text).

\end{module}


\section{{\bell} {\privatebit}s}
\label{sec:privatestates}

\begin{module}
    \newcommand{\bp}{\{p_k\}}
    The {\eprbit}s produce {\secretbit}s with respect to the environment
    when measured in the computational basis.
    The {\privatebit}s are those states that generalize this property,
    namely they are all those states that produce {\secretbit}s with respect to the environment
    when measured in the computational basis.
    Let us recall that a state is a {\privatebit} if and only if
    it can be written as\cite{secure}:
    \[\gamma^m \coloneqq \twisting(\Phi\otimes \sigma)\twisting^\dagger\]
    where $m=\log|\keyA|=\log|\keyB|$,  
    $\twisting$ is a controlled unitary 
    $\twisting = \sum _{ij} \proj{ij} \otimes U_{i}$, 
    and $\sigma$ is an arbitrary state.
    It is implicit that $\Phi$ is a state of $\key$ and $\sigma$ is a state of $\shield$.
    Let us also recall the definition of Bell private states: 
    \[\gamma_\bell^m \coloneqq \sum_{\kindex\in\integers_{2^m}}
    p_\kindex \, \phi _{0\kindex} \otimes \sigma _\kindex\]
    where $\sigma _\kindex$ are arbitrary orthogonal states.
    Again $\phi_{ij}$ are states of $\key$ and $\sigma_\kindex$ are states of $\shield$.

    \begin{module}
        All {\bell} {\privatebit}s are {\privatebit}s.
        This can be  proved either
        by checking that the measurement in $\key$ gives perfectly secure bits
        or by showing that they admit an expression as {\privatebit}s, here we show the latter.
        \begin{lemma}%
            \label{lemma:bell-twisting}
            A state $\gamma_\bell^m$ is a {\bell} {\privatebit}
            if and only if 
            it can be written as $\gamma_\bell^m=\twisting(\Phi\otimes \sigma)\twisting^\dagger$ 
            using 
            \begin{align*}
            \twisting &= \sum \nolimits_{ij} \proj{ij} \otimes U_\sigma^i
            &
            \sigma &= \sum \nolimits_\kindex p_\kindex \sigma _\kindex
            \intertext{where $U_\sigma^i$ is the $i$-th power of } 
            U_\sigma &= \sum \nolimits_\kindex \omega^\kindex P_{\sigma _\kindex} + P_{\sigma _\perp}
            \end{align*}
            with $P_{\sigma _\kindex}$ the projectors onto the supports of $\sigma _\kindex$ and
            $P_{\sigma _\perp} = \one -\sum_\kindex P_{\sigma_\kindex}$ the remaining orthogonal projector.
        \end{lemma}
        $P_{\sigma _\perp}$ plays no active role,
        it is only needed to complete $\twisting$,
        so that $\sigma$ is not required to have full support.
        \begin{proof}
            \providecommand*{\phagger}{{\phantom\dagger}}
            The following sequence of equalities proves that a state is of the form 
            $\sum p_\kindex \, \phi _{0\kindex} \otimes \sigma_\kindex$,
            thus a {\bell} {\privatebit},
            if and only if it is 
            of the form $\twisting (\Phi^\phagger \otimes \sigma) \twisting^\dagger $
            with  $\sigma$ and $\twisting$ as above.
            \allowdisplaybreaks%
            \renewcommand{\kindex}{k}
            \begin{align*}
            \gamma^m
            &=\sum_k
            p_\kindex \phi _{0\kindex} \otimes \sigma _\kindex
            \\&=\sum_{ijk}
            p_\kindex \frac1{2^m}
            \omega^{i\kindex} \ketbra{ii}{jj} \omega^{-j\kindex}
            \otimes P_{\sigma _\kindex} \sigma _\kindex P_{\sigma _\kindex}
            \\&=\sum_{ijk}
            \frac{1}{2^m}\ketbra{ii}{jj}  \otimes
            (\omega^{ i\kindex}P_{\sigma _\kindex})
            \cdot(p_\kindex\sigma _\kindex)\cdot
            (\omega^{-j\kindex}P_{\sigma _\kindex})
            \\&=\sum_{ijk\alpha\beta}
            \frac{1}{2^m}\ketbra{ii}{jj} \otimes
            (\omega^{ i\alpha}P_{\sigma _\alpha})
            \cdot (p_\kindex\sigma _\kindex) \cdot
            (\omega^{-j\beta}P_{\sigma _\beta})
            \\&=\sum_{ij}
            \frac{1}{2^m}\ketbra{ii}{jj}
            \otimes\\&\qquad\otimes
            {({\textstyle\sum_\alpha}\omega^\alpha P_{\sigma _\alpha}+ P_{\sigma _\perp})}^i
            \cdot\sigma\cdot
            {({\textstyle\sum_\beta}\omega^\beta  P_{\sigma _\beta }+ P_{\sigma _\perp})}^{-j}
            \\&=\sum_{ij}
            \frac{1}{2^m}\ketbra{ii}{jj}  \otimes
            U_\sigma^i
            \sigma
            U_\sigma^{-j}
            \\&=\sum_{ij}
            (\proj i \otimes U_\sigma^i)     \cdot
            \quantity( \tfrac{1}{2^m}\ketbra{ii}{jj} \otimes\sigma )
            \cdot (\proj j \otimes U_\sigma^{j\dagger})
            \\&= \twisting^\phagger ( \Phi^m \otimes \sigma ) \twisting^\dagger
            \end{align*}
            where we used the orthogonality of the $\sigma _\kindex$ in
            the identity
            \[\sum_{\kindex\alpha}P_{\sigma _\alpha} \sigma _\kindex
            = \sum_\kindex P_{\sigma _\kindex} \sigma _\kindex
            \;.\qedhere\]
        \end{proof}
    \end{module}

    \begin{module}
        \providecommand{\X}{Y}
        
        For $|\keyA|=|\keyB|=2$,
        any {\privatebit} admits a \emph{block form}~\cite{paradigms},
        namely it can be written as:
        \[ \gamma^1 = \frac{1}{2}
        \begin{pmatrix}
        \sqrt{\X^\dagger \X} & 0 & 0 & \X^\dagger
        \\0&0&0&0\\0&0&0&0\\
        \X & 0 & 0 & \sqrt{\X \X^\dagger}
        \end{pmatrix}\]
        where $\X$ is any opportune matrix of unit trace norm  ($\norm{\X}_1 = 1$).
        This can be easily seen by recalling
        that any matrix $\X$ admits a singular value decomposition
        and noticing that the decomposition can be used to extract $\sigma$
        and the unitaries $U_0$ and $U_1$ in $\twisting$.
        However, this does not work in higher dimension ($m=\log|\keyA|>1$)
        because then additional unitaries are needed to specify $T$.
        
        This is not true for {\bell} {\privatebit}s, indeed,
        a {\bell} {\privatebit} only needs to specify a single unitary $U_\sigma$.
        This allows one to write a block form for all {\privatebit}s
        by exploiting the fact that $U_\sigma$ and $\sigma$ commute.
        \begin{corollary}%
            \label{lemma:bellpdits-blockform}
            $\bgamma^m$ is a {\bell} {\privatebit} iff
            \[ \bgamma^m =\frac1{2^m} \sum_{ij} \ketbra{ii}{jj}
            \otimes |\X| \quantity( \frac{\X}{|\X|}) ^{i-j}\]
            for some normal $\X$ ($\X^\dagger\X=\X\X^\dagger$) such that $\norm{\X}_1 = 1$ and $\X^{2^m} \geq 0$.
        \end{corollary}
        \begin{proof}
            Set $Y = \sigma U_\sigma$.  Then $|Y|=\sigma$, $\frac{Y}{|Y|} = U_\sigma$ and
            \[U_\sigma^i \sigma U_\sigma^{-j} = \sigma U_\sigma^{i-j}
            = |Y| \quantity(\frac{Y}{|Y|})^{i-j}\;.\qedhere\]
        \end{proof}
        \noindent In \Cref{lemma:bellpdits-blockform},
        $\X^{-1}$ is intended as pseudo inversion of matrices
        so that $\X$ need not to be full rank.
        Note how in the corollary $m=1$ the corollary implies  $\X=\X^\dagger$.
    \end{module}

    \begin{module}
        \providecommand{\bp}{\mathbf{p}}
        \providecommand{\obgamma}{\mathring\gamma_{{\bell}}}
        \providecommand{\hbgamma}{\hat\gamma_{{\bell}}}
        \providecommand{\all}{{AB}}
        We now consider some simple entropic properties of {\privatebit}s.
        For this purpose let us introduce the following state
        \[\mathring\gamma = \gamma _\key \otimes \gamma _\shield\]
        namely the tensor product of the key and shield marginals of a {\privatebit}.
        For {\bell} {\privatebit}s this has a form similar to the {\keyattacked}:
        \[\obgamma = \sum \nolimits_\kindex p_\kindex \phi _{0\kindex} \otimes \sigma\]
        which indeed gives $\obgamma = \hbgamma$ for uniform probability distributions $\bp$.
        
        We can summarize the difference between $\mathring\gamma$
        and the {\keyattacked} $\hat\gamma$
        using quantum relative entropies ($D(\rho\parallel\sigma) = \tr[\rho\log\rho-\rho\log\sigma]$)
        as follows:
        \begin{align*}
        D(\bgamma^m\|\hbgamma^m)
        &= D(\Phi \|\hat\Phi)
        = m
        \\
        D(\bgamma\|\obgamma)
        &= I(\key{:}\shield)_{\bgamma}
        = H(\bp)
        \end{align*}
        where
        $ I(\key{:}\shield)$ is the quantum mutual information
        and $H(\bp)$ is the entropy of $\bp$.
        Since $\bgamma$ commutes with both $\hbgamma$ and $\obgamma$,
        these values are achieved also
        by performing a global measurement first
        and then computing the classical relative entropies~\cite{petz,variational}.
        In both cases the optimal measurement  operators are
        $\{M_{ab\kindex}\} = \{P_{\phi _{ab}} \otimes P_{\sigma _\kindex}\}$.
    \end{module}

    \begin{module}
        \providecommand{\bp}{{\mathbf{p}}}
        Another simplification of {\bell} {\privatebit}s with respect to
        general {\privatebit}s involves the expression
        for the distillable entanglement in $\key$:
        \begin{lemma}%
            \label{lemma:keymarginal}
            For all {\bell} {\privatebit}s $\bgamma^m$ it holds
            \[E_D(\gamma ^m _{\bell,\key}) = m - H(\bp)\]
            where $\gamma ^m _{\bell,\key} = \tr _\shield \bgamma^m$.
        \end{lemma}
        \begin{proof}
            For all {\privatebit}s, the reduced state of $\key$ has support only
            on the maximally correlated subspace.
            For such reduced states,it has been shown~\cite{e-maxcorr}
            that the distillable entanglement is equal to the hashing bound~\cite{dw}:
            \[E_D(\tr_\shield \gamma)
            = H(\keyB)_{\gamma} - H(\key)_{\gamma}\;.\]
            For {\bell} {\privatebit}s we have
            \[\tr_\shield \bgamma^m = \sum \nolimits_\kindex p_\kindex\phi _{0\kindex}\;.\]
            The marginal of $\keyB$ will be completely mixed
            so $H(\keyB)_{\gamma}$ will be maximal,
            while $H(\key)_{\gamma}$ is  the entropy of $\bp$:
            \[E_D(\tr_\shield \bgamma^m) = m - H(\bp)\;. \qedhere\]
        \end{proof}
        \noindent
        Of particular interest is the case of uniform $\bp$, then:
        \[E_D(\tr_\shield \bgamma)=0\]
        which is independent of the key systems size.
    \end{module}

\end{module}


\section{{Relative Entropies}}

\begin{module}
    \providecommand{\measurements}{{\mathbb{M}}}
    \providecommand{\maps}{{\mathbb{L}}}
    \providecommand{\qmap}{{\Lambda}}
    \providecommand{\states}{{P}}
    \providecommand{\state}{{\sigma}}
    The \emph{quantum relative entropy} between two states $\rho$ and $\sigma$ is defined as:
    \[D(\rho\parallel\sigma) \coloneqq
    \tr \rho[\log\rho - \log\sigma] 
    \;.\]
    Some interesting entanglement measures are defined using the relative entropy.
    %
    \begin{definition}
        [\cite{relentropy}, Relative entropy of entanglement]
        \[E_R(\rho) \coloneqq \inf _{\sigma \in SEP} D( \rho \parallel \sigma )\]
    \end{definition}
    
    \noindent
    This was generalized as follows.
    
    \begin{definition}
        [\cite{piani}, Relative entropy with respect to $\bm \states$]
        \label{def:relent-P}
        Let $\states$ be any set of states $\states$ containing at least one full rank state.
        Define:
        \[E^\states_R(\rho) \coloneqq \inf _{\sigma \in \states} D( \rho \parallel \sigma )\;.\]
    \end{definition}
    
    \noindent
    By taking the relative entropy with respect to separable states,
    we recover the relative entropy of entanglement.
    The relative entropy with respect to $P$ is asymptotically continuous~\cite{relent-continuity},
    a useful property needed to prove our bounds on rates.
    We will restate asymptotic continuity here in the form we need it,
    using the improved bounds from~\cite[Lemma 7]{tight-continuity}.
    
    \begin{lemma}[\cite{tight-continuity}]%
        \label{lemma:continuity}
        Let $C=\sup_\rho E_R^\states(\rho)$. Then:
        \[|E_R^\states(\rho) - E_R^\states(\sigma) |
        \leq {C}\eps - \eta(\eps) + \eta(1+\eps) \]
        where $\eps=\frac{1}{2}\norm{\rho-\sigma}_1$ and $\eta(x) = x\log x$.
    \end{lemma}
    In particular for $\states\supseteq SEP_{A:B}$ and $d=\min\{|A|,|B|\}$,
    we can always take $C=\log d$.
    
    For an arbitrary set of measurements $\measurements$
    we define the \emph{$\measurements$-relative entropy}~\cite{piani}.
    However, for now we allow any set of quantum maps $\maps$
    (completely positive trace preserving maps) in the definition,
    as opposed to~\cite{piani} where $\maps$ only contains measurements.
    
    \begin{definition}[\cite{piani}, $\maps$-relative entropy]%
        \label{def:M-relent}
        Let $\maps$ be any set of quantum maps. Define
        \[D_\maps (\rho\parallel\sigma)
        = \sup _{\Lambda \in \maps}
        D(\Lambda(\rho)\parallel \Lambda(\sigma))\;.\]
    \end{definition}
    Originally the above relative entropy
    was defined only for sets of measurements $\measurements$
    However, for now we will allow any set of quantum maps
    (completely positive trace preserving maps) in the definition,
    as opposed to~\cite{piani} where $\maps$ only contains measurements.
    Combining \Cref{def:relent-P,def:M-relent} we obtain the following definition.
    
    \begin{definition}
        [\cite{piani}, $\maps$-relative entropy with respect to~$\states$]%
        \label{def:M-relent-P}
        Let $\maps$ be any set of quantum maps
        and let $\states$ be any set of states $\states$ containing at least one full rank state. 
        Define
        \[E _{R,\maps}^\states (\rho) =
        \inf _{\sigma \in \states} D_\maps (\rho\parallel\sigma)\;.\]
    \end{definition}

    
    \setlength{\tabcolsep}{6pt}
    \begin{table*}[p]
        \begin{tabular}{lll} 
            \toprule
            \multicolumn{2}{c}{Notation} & \multicolumn{1}{c}{Meaning}
            \\ \cmidrule(r){1-2}
            Full & Short & 
            \\ \midrule
            $\key$& & The key systems of Alice and Bob, $|\keyA|=|\keyB|$.
            We define $m=\log_2 |\keyA|$.
            \\ \midrule 
            $\shield$ & 
            & {The shield systems of Alice and Bob.
                In most the examples we use $|\shieldA|=|\shieldB|=d$.}
            \\ \midrule 
            $\rho,\hat\rho$ & 
            & A state on $\key\shield$ and its {\keyattacked} 
            \\ \midrule 
            $\measurement$ &
            & A measurement.
            \\ \midrule 
            $\Lambda$ &
            & A map/protocol.
            \\ \midrule 
            $LO_A, LO_B$ &
            & {Local operations: any quantum channel acting only
                on Alice's or Bob's systems.}
            \\ \midrule 
            \specialcell[c]{$\locc_\ab$\\$\locc_\atb$} &
            & Local Operations with two-way
            and one-way Classical Communication.
            \\ \midrule 
            $\measurement \in LO_A$ & $M_A$ 
            & A measurement at Alice's side.
            \\ \midrule 
            \specialcell[c]{$\measurement\in\locc_\ab$\\$\measurement\in\locc_\atb$}
            & \specialcell[c]{$\measurement_\ab$\\$\measurement_\atb$}
            & A measurement in $\locc_\ab$ and $\locc^\to_\ab$ respectively.
            \\ \midrule 
            $\measurement\in\locc_{\measurements\ab}$
            & $\measurement_{\measurements\ab} $ 
            & A partial measurement in $\locc_\ab$,
            Alice needs to measure, but not Bob.
            \\ \midrule 
            \specialcell[c]{$\Lambda\in\locc_\ab$\\$\Lambda\in\locc_\atb$}
            & \specialcell[c]{$\Lambda_\ab$\\$\Lambda_\atb$}
            & A protocol in $\locc_\ab$ and $\locc_\atb$ respectively
            \\ \midrule 
            $D_A$ & & \specialcell{Relative entropy restricted to partial measurements
                $\maps = \measurements \otimes \id_{A}$,
                where $\measurements=LO_A$.}
            \\ \midrule 
            $D_\atb$ & 
            & Relative entropy restricted to measurements in $\measurements=\locc_\atb$.
            \\ \midrule 
            $E_{R,\atb}$ & 
            & Relative entropy with measurements
            in $\measurements=\locc_\atb$ and $\states=SEP_\ab$.
            \\ \bottomrule
        \end{tabular}
        \stepcounter{figure}
        \caption{\label{table:notation}%
            General notation for states, maps, measurements, parties, etc.}
    \end{table*}
    
    \setlength{\tabcolsep}{4pt}
    \begin{table*}[p]
        \repeatersystems%
        \begin{tabular}{llp{0.718\linewidth}}
            \toprule
            \multicolumn{2}{c}{Repeater notation} & \multicolumn{1}{c}{Meaning}
            \\ \cmidrule(r){1-2}
            Full & Short & 
            \\ \midrule
            $\charlieA\charlieB$ & $\charlie$
            & {The parties on Charlie's side.
                $\charlieA$ and $\charlieB$ share entanglement 
                with Alice and Bob respectively.}
            \\ \midrule 
            $\keyof\charlieA\keyof\charlieB$ & $\keyof\charlie$
            & The key systems of Charlie's parties.
            \\ \midrule 
            $\shieldof\charlieA\shieldof\charlieB$ & $\shieldof\charlie$
            & The shield systems of Charlie's parties.
            \\ \midrule 
            $\sigma\in\sep_\accb$ & $\sigma_\accb$
            & Quadri-separable states of Alice, Bob, and Charlie's parties.
            \\ \midrule 
            $\sigma\in\sep_\repcutA$ & $\sigma_\repcutA$ 
            & Separable states between Alice
            and the joint systems of Charlie and Bob.
            \\ \midrule 
            $\sigma\in\sep_\repcutB$ & $\sigma_\repcutB$
            & Separable states between Bob and the joint systems of Charlie and Alice.
            \\ \midrule 
            $\locc_\acb$ & 
            & Tripartite LOCC of Alice, Charlie and Bob.
            \\ \midrule 
            $\locc_\ctab$ & 
            & {Tripartite LOCC with only one-way communication from Charlie to Alice/Bob.
                Alice and Bob can communicate freely with each other.}
            \\ \midrule 
            $\measurement\in LO_\charlie$ & $\measurement_\charlie$ 
            & A partial measurement at Charlie's.
            \\ \midrule 
            \specialcell[c]{$\measurement\in\locc_\acb$\\$\measurement\in\locc_\ctab$} 
            & \specialcell[c]{$\measurement_\acb$\\$\measurement_\ctab$}
            & A measurement in $\locc_\acb$ and $\locc_\ctab$ respectively.
            \\ \midrule 
            \specialcell[c]{$\Lambda\in\locc_\acb$\\$\Lambda\in\locc_\ctab$} 
            & \specialcell[c]{$\Lambda_\acb$\\$\Lambda_\ctab$}
            & A protocol in $\locc_\acb$ and $\locc_\ctab$ respectively.
            \\ \midrule 
            $D_\charlie$ & 
            & {Relative entropy restricted to partial measurements
                $\maps = \measurements \otimes \id_{AB}$,
                where $\measurements=LO_\charlie$.}
            \\ \midrule 
            $D_\ctab$ & 
            & Relative entropy restricted to measurements in
            $\measurements= \locc_\ctab$.
            \\ \midrule 
            $E_{R,\ctab}^{\accb} $ & 
            & Relative entropy with $\states=\sep_\accb$
            and measurements in $\measurements=\locc_\ctab$.
            \\ \midrule 
            \specialcell[c]{$E_{R,\ctab}^{\repcutA}$\\$E_{R,\ctab}^{\repcutB}$} & 
            & {Relative entropy with $\states=\sep_\repcutA,\sep_\repcutB$
                and measurements in $\measurements=\locc_\ctab$.}
            \\ \midrule 
            \specialcell{$E_{R,\measurements\cab}^\repcutA$\\$E_{R,\measurements\cab}^\repcutA$} &
            & {Relative entropy with $\states=\sep_\repcutA,\sep_\repcutB$
                and $\maps = \locc_{\measurements\ctab}$,
                which are LOCC protocols of $\cab$ followed by a measurement at Charlie.}
            \\ \bottomrule
        \end{tabular}
        \stepcounter{figure}
        \caption{\label{table:notation-repeater}%
            Additional notation regarding the repeater setting.} 
    \end{table*}

    
    \section{{Regularized relative entropies}}

    Given a function on states $f$
    we can consider the standard regularization
    \[f^\infty(\rho) = \lim_{n\to\infty} \frac{1}{n} f(\rho^{\otimes n})\]
    whenever the limit is well defined.
    Similarly for functions $g$ of two states
    we can consider a regularization
    \[g^\infty(\rho,\sigma) = \lim_{n\to\infty} \frac{1}{n} g(\rho^{\otimes n}, \sigma^{\otimes n})\]
    again whenever the limit is well defined.
    
    Fekete's lemma~\cite{fekete} guarantees that the regularization will be well defined,
    at least for the classes of partial measurements that we consider
    in \Cref{lemma:datahiding-oneway} and
    \Cref{lemma:repeatersep-oneway:appendix}.
    We call a sequence $D_n$ \emph{super-additive} if it satisfies
    \begin{equation}
    \forall n,m: D_{m+n} \geq D_n+D_m
    \;.
    \label{eq:monotonicity-sum}
    \end{equation}
    The lemma states that if a sequence $D_n$ is super-additive,
    then $\frac{1}{n} D_n$ either converges or diverges to infinity,
    more specifically
    $\lim_{n\to\infty} \frac{1}{n}D_n
    = \sup_{n\in\naturals} \frac1{n}{D_n}$.

    \begin{module}
        \newcommand{\hilbert}{{\mathcal{H}}}
        \newcommand{\density}[1]{{\mathcal{D} (#1)}}
        \noindent
        Fekete's lemma guarantees that we can regularize
        the various relative entropies defined above,
        as long as $\maps$ and $\states$
        make these relative entropies super-additive.
        This is the case for classes of maps like $\locc$
        that are defined for states $\rho$ and $\rho^{\otimes n}$ alike,
        and that map separable states to separable state.
        To make this rigorous we need to define $\maps$ and $\states$
        for every $n$ and therefore we need to
        explicitly consider the Hilbert space $\hilbert$
        and the set of density matrices $\density\hilbert$ for $n=1$;
        let $\hilbert$ be finite dimensional.
        
        We say that $\maps$ is a \emph{class closed under tensor products},
        or simply \emph{class}, of quantum maps (or measurements)
        if it is a sequence of $\maps_n$,
        each a set of quantum maps on $\density{\hilbert\tensor n}$, such that: 
        \[\forall\ i,j,
        \ \Lambda\in\maps_i,
        \ \Gamma \in\maps_j:
        \quad
        \Lambda\otimes\Gamma\in\maps_{i+j}
        \;.\]
        In such case it is easy to check that
        $D_{\maps_{n}}(\rho^{\otimes n}\|\sigma^{\otimes n})$
        satisfies \Cref{eq:monotonicity-sum},
        \begin{align*}
        D_{\maps_{i+j}} (\rho\tensor{i+j}\|\sigma\tensor{i+j})
        \geq D_{\maps_{i  }} (\rho\tensor{i  }\|\sigma\tensor{i  })
        +   D_{\maps_{  j}} (\rho\tensor{  j}\|\sigma\tensor{  j})\;,
        \end{align*}
        and thus we can define the following regularization.
        
        \begin{definition}[Regularized $\maps$-relative entropy]%
            \label{def:reg-M-relent}
            Let $\maps$ be a class of quantum maps closed under tensor products.
            Then, define:
            \[D_\maps ^\infty (\rho\parallel\sigma) =
            \lim _{n\to\infty} \smash{\frac{1}{n} }
            D_{\maps_n}(\rho^{\otimes n}\parallel \sigma^{\otimes n})\;.\]
        \end{definition}

        We now restrict $\maps$ to be a class of $\measurements$ measurements only,
        so that we can regularize $D_\measurements^\states$.
        As done in~\cite{piani}, we also need to impose further conditions on $\states$
        with respect to $\measurements$ and on $\states$ itself.
        
        We say that $\states$ is a \emph{class closed under partial trace},
        or simply \emph{class}, of states,
        if it is a sequence of $\states_n$,
        each a convex set of states in $\density{\hilbert\tensor n}$,
        such that for all $n$:
        \[\sigma\in \states_n
        \quad\Rightarrow\quad
        (\tr_I \sigma \in \states_{n-|I|}
        \quad\forall\ I\subset\{1\dots n\})\]
        where $\tr_I$ is the partial trace over
        the systems designated by the index set $I$.
        Furthermore, let $\measurements$ be a class of measurements,
        we say that a class of states $\states$ is \emph{closed under $\measurements$},
        if for all $n>m$, it holds that 
        \[\begin{cases}
        \sigma\in \states_n
        \\
        \measurement\in\measurements_m
        \end{cases}
        \Rightarrow\quad
        \frac{\tr_I M^I_k \sigma}{\tr M_k^I \sigma} \in \states_{n-m}
        \quad\forall\ k,|I|=m
        \]
        where $\{M_k\}$ are the measurements operators of $\measurement$
        and $M^I_k$ are the operators acting on systems $I$
        ($M_k^I \equiv M_k^I \otimes \one$, 
        where the identity acts on the complement of $I$).
        These properties are enough to guarantee
        that the $\measurement$-relative entropy with respect to $\states$
        also satisfies \Cref{eq:monotonicity-sum}.
        While the original statement only considers separable states,
        PPT states, separable measurements and LOCC measurements,
        the same exact proof carries over to general $\states$ and $\measurements$
        
        \begin{lemma}
            [\protect{\cite[Theorem 2(d)]{piani}}]
            Let $\measurements$ be a class of measurements
            and let $\states$ be a class of states closed under $\measurements$.
            Then:
            \[E_{R,\measurements_{m+n}}^\states(\rho\tensor{m+n})
            \geq E_{R,\measurements_{m}}^\states(\rho\tensor{m})
            +  E_{R,\measurements_{n}}^\states(\rho\tensor{n})\;.\]
        \end{lemma}
        
        
        Thanks to this lemma the following regularization is now well defined.
        
        \begin{definition}%
            [Regularized $\measurements$-relative entropy with respect to $\states$]%
            \label{def:reg-M-relent-P}
            Let $\measurements$ be a class of measurements
            and let $\states$ be a class of states closed under $\measurements$.
            Define:
            \[E _{R,\measurements}^{\infty,\states} (\rho) =
            \lim _{n\to\infty} \smash{\frac{1}{n} }
            E _{R,\measurements}^\states(\rho^{\otimes n})\;.\]
        \end{definition}
    \end{module}

    Notice that the class of PPT states
    is closed under the class of PPT measurements
    and consequently under all subclasses of measurements,
    like separable and LOCC measurements.
    Similarly the class of separable states
    is closed under the class of separable measurements,
    and consequently under LOCC measurements.
    Thus \Cref{def:reg-M-relent-P} is always well defined for
    the above combinations of states and measurements.
    
    As it is usually done, we will omit the fact that
    the classes of states and quantum maps/measurements
    are actually sequences in such regularized quantities,
    therefore we will drop the index $n$.
    We will use $\sigma_\states$ and $\measurement_\measurements$
    as short hand notation for $\sigma\in\states$
    and $\measurement\in\measurements$, respectively.
    See \Cref{table:notation} for
    a more detailed list of symbols and notations.
    
    Finally, we highlight that we do not have yet
    a general \Cref{def:reg-M-relent-P}
    for arbitrary classes of quantum maps.
    This in particular includes classes of partial measurements,
    where only some parties are forced to measure their systems.
    Already for partial measurements,
    the only regularized definition that we can use for now
    is \Cref{def:reg-M-relent}.
    
    %
\end{module}%


\begin{figure}
    \renewcommand{\achievedrate}{E}
    \newcommand{\circuitC}{.85em}
    \newcommand{\circuitR}{.5em}
    
    \begin{module}
        \providecommand{\circuitC}{.7em}
        \providecommand{\circuitR}{.3em}
        \providecommand{\circuitRfix}{1em}
        \providecommand{\target}{\Phi^{n\achievedrate}}
        \providecommand{\dw}{\ar@{.} [0,-1]&\ar@{.} [0,-1]} 
        
        \[
        \underbrace{\rho\otimes \dots \otimes \rho}_{n\text{ copies}}\hspace{0.25em}
        {\begin{gathered}\Qcircuit  @C=\circuitC @R=\circuitR @!R=\circuitRfix {
                &\ustick{\scriptstyle  {Alice}}
                &\qw&\multigate{1}\empty         &&&&\qw[-4] &\dw&\multigate{1}\empty&\qw
                \\
                &&&\pureghost\empty&\controlo\cw&&&&&\controlo&\pureghost\empty\cw
                \\\ar @{-} [-2,1] \ar @{-} [2,1]&&&&&&\scriptstyle\distillmaps&&&\\
                &&&&&\controlo\ar@{=}[-2,-1]&\pureghost\empty\cw&\controlo\cw&&&
                \\
                &\dstick{\scriptstyle  {Bob}  }
                &&&\push{\rule{0em}{1em}}&\qw[-4]&\multigate{-1}\empty&\qw&\dw&\push{\rule{0em}{.9em}}&\qw[-2]
                \gategroup{1}{4}{5}{11}{.5em}{--}
                \save "2,9"."3,8"  \ar@{:}"4,8" \restore 
                \save "4,9"."3,10" \ar@{:}"2,10"\restore 
            }\end{gathered}}
        \;\;\approx _\eps \;
        \target\hspace{0.25em}
        \begin{gathered}\Qcircuit  @C=\circuitC @R=\circuitR @!R=\circuitRfix {
            &\ustick{\scriptstyle  {Alice}}&\qw
            \\\\\ar @{-} [-2,1] \ar @{-} [2,1]\\\\
            &\dstick{\scriptstyle  {Bob}  }&\push{\rule{0pt}{1.075em}}\qw
        }\end{gathered}
        \]
    \end{module}
    \caption{\label{figure-distillation}%
        Quantum circuit for entanglement distillation.
        The circuit in the dashed box is the map $\Lambda\in\locc_\ab$,
        optimized to output $\Phi^{n\achievedrate}$
        with the highest $n\achievedrate$ possible.}
\end{figure}
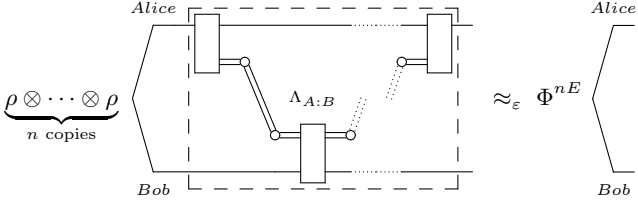


\section{{Entanglement and Key distillation}}
\label{sec:distillation-rates}

\begin{module}
    
    When performing distillation,
    the goal is to approximate a desired output state
    by acting on the input via the allowed operations.
    Namely, four elements define a rate:
    the input state, the target states, the allowed operations
    and a measure of approximation.
    Beyond this, all the rates have a common structure.
    Usually, trace norm distance is used to quantify the approximation between two states.
    We will use the notation
    \[\varrho \approx_\eps \varsigma
    \qquad\Leftrightarrow\qquad
    \frac{1}{2} \norm{\varrho - \varsigma}_1 \leq \eps\;.\]
    In what follows we will use the same convention as in~\cite{limits}
    in the definitions of rates, namely all rates are defined as a single formula.
    
    \begin{module}
        \renewcommand{\achievedrate}{E}
        The distillable entanglement of a state $\rho$
        is then formally defined as the \textrate{{\eprbit}s}
        under bipartite LOCC~\cite{logneg}:
        \[E_D(\rho_{AB}) =\distillrate{\rho_{AB}}{\Phi_{AB}}\]
        as illustrated in \Cref{figure-distillation}.
    \end{module}
    
    The distillable key is defined as the \textrate{{\secretbit}s}.
    However it also equals the \textrate{{\privatebit}s}
    under bipartite LOCC~\cite{secure}:
    \[K_D(\rho) = \distillrate\rho\gamma\;.\]
    Note how the expression is almost
    the same as for distillable entanglement,
    the only difference are the desired output states.
    Using the reversible map
    of \Cref{lemma:reversible-rho} 
    it is now also possible to write the distillable key
    as the \textrate{{\bell} {\privatebit}s}:
    \[K_D(\rho) = \distillrate\rho\bgamma\;.\]
    
    \smallskip
    While the target state defines the kind of resource being measured
    (pure entanglement, key, \dots),
    changing the available protocols produces variations of these quantities
    that reflect different scenarios, like in the case of the 
    %
    one-way distillable entanglement.
    This is:
    \renewcommand\achievedrate{E}
    \renewcommand\distillmaps{\distillmap_\atb}
    \[E_D^\to(\rho) =\distillrate\rho\Phi\]
    where the maps are restricted to one-way $\locc$.
    Just like there is a regularized $\locc$ restricted
    relative entropy of entanglement upper bound on $E_D$,
    the same proof carries over to $E_D^\to$ by simply restricting to one-way $\locc$.
    Here below, we give an explicit proof of such bound.
    
    \begin{lemma}[\cite{Emeasures}]%
        \label{lemma:regularized-1Er}
        For any state $\rho$ and any separable state $\sigma$, it holds:
        \[D_A^\infty(\rho\parallel\sigma)
        \geq E_{R,\atb}^\infty(\rho) 
        \geq E_D^\to(\rho)
        \;.\]
    \end{lemma}
    It should be understood that the direction of the communication in the measurement
    must be the same as the direction in the distillation protocol.
    \begin{proof}
        The proof of $E_{R,\atb}^\infty(\rho) \geq E_D^\to(\rho)$
        mirrors the same result for two-way LOCC from~\cite{Emeasures}.
        We need two results from the same article,
        the value of $E_{R,\atb}$ on {\eprbit}s~\cite[Proposition 4]{Emeasures}:
        \begin{align*}
        E_{R,\atb}(\Phi^m) &= \log (2^m +1) -1 \\&= m + \log (1 + 2^{-m}) -1 \;,
        \end{align*}
        {and its asymptotic continuity~\cite[Proposition 3]{Emeasures}}:
        \[\quantity| E_{R,\atb}(\varrho) - E_{R,\atb}(\varsigma) |
        \leq 4\eps\log{\frac{3|AB|}{\eps}}\]
        for states satisfying $\frac{1}{2}\norm{\varrho-\varsigma}_1 = \eps\leq e^{-1}$.
        Now, for all $\eps$ and for all $n$ let
        $\Lambda$ be the optimal one-way distillation map of $E_D^\to(\rho)$.
        Then: 
        \begin{align*}
        \frac{1}{n} &E_{R,\atb}(\rho^{\otimes n})
        \\
        &\geq \frac{1}{n}
        E_{R,\atb}({\Lambda(\rho^{\otimes n})})
        \\
        &\geq \frac{1}{n}
        E_{R,\atb}({\Phi^{n\achievedrate}})
        - \frac{1}{n} 4\eps \log \frac{3 \,2^{2n\achievedrate}}\eps
        \\
        &=    E + \frac{1}{n} \log\quantity(1+\frac{1}{2^{n\achievedrate}})
        - \frac{1}{n}
        - 8\eps \achievedrate - \frac{4\eps}{n} \log \frac{3}{\eps}
        \end{align*}
        Taking the limit $n\to\infty$ leads to:
        \[ E_{R,\atb}^\infty(\rho) \geq \achievedrate - 8\eps \achievedrate\]
        and taking the limit $\eps\to 0$ ends the first part of the proof.
        
        We now prove
        $D_A^\infty(\rho\parallel\sigma)\geq D_\atb^\infty(\rho\parallel\sigma)$.
        This follows because the optimization in $D_A$
        is made over a larger class of maps,
        because Bob does not necessarily measure.
        More precisely, any measurement in $\measurement\in\locc_\atb$,
        can always be written as a measurement $\measurement'\in LO_A$  followed
        by a measurement on Bob conditioned on Alice's outcome~\cite{dw},
        namely a global measurement $\measurement''$ acting on Bob
        and Alice's measurement outcome.
        The communication is implicit in the fact
        that $\measurement''$ has received the outcome of $\measurement'$
        and is treating $\measurement'(\rho)$ as a global state.
        By the monotonicity of the relative entropy we thus have:
        \begin{align*}
        D(\measurement'(\rho)\parallel \measurement'(\sigma))
        &\geq D(\measurement''\circ \measurement'(\rho)
        \parallel \measurement''\circ \measurement'(\sigma))
        \\&= D(\measurement(\rho)\parallel \measurement(\sigma)) \;.
        \end{align*}
        Thus:
        \[D_A(\rho\parallel\sigma) \geq D_\atb(\rho\parallel\sigma) \]
        and therefore
        \[D_A^\infty(\rho\parallel\sigma) \geq D_\atb^\infty(\rho\parallel\sigma) \;.\]
        Since by definition $D_\atb^\infty(\rho\|\sigma)\geq E_{R,\atb}^\infty(\rho)$,
        this concludes the proof.
    \end{proof}

\end{module}


\section{{Distillable entanglement of {\privatebit}s}}

\providecommand{\zkindex}{{0\kindex}}
%
%
In this Appendix we provide a rigorous proof of \Cref{lemma:datahiding-oneway}.
We start by proving the in-line claims of the main text.

\begin{lemma}%
    \label{lemma:datahiding-belldiagonal}
    Let $\rho$ be a key correlated state of the form
    $\rho = \frac{1}{2^m} \sum _\kindex \phi_\zkindex \otimes \sigma_\kindex$.
    Then it holds:
    \begin{align*}
    E_D^\to(\rho) 
    &\geq \sup  _{\measurement \in LO_A} \frac{1}{2^m} \sum _\kindex
    D(\measurement(\sigma _\kindex)\parallel\measurement(\sigma))
    \end{align*}
    where $\sigma = \sum _\kindex \frac{1}{2^m}  \sigma_\kindex$,
\end{lemma}

\begin{proof}
    \newcommand{\rhomeasured}{\tilde\rho}
    Let Alice perform a measurement $\measurement:\shieldA\!\to M$
    on her shield and define:
    \begin{equation}
    \rhomeasured \coloneqq
    {\frac{1}{|\keyB|} \sum _\kindex}
    \phi _\zkindex \otimes \measurement(\sigma _\kindex)\;.
    \end{equation}
    We have $E_D^\to(\rho)\geq E_D^\to(\rhomeasured)$.
    Alice now sends the outcome to Bob and
    then by the hashing bound~\cite{dw} we find:
    \[E_D^\to(\rhomeasured) \geq 
    [H(\keyB\shieldB M)_{\rhomeasured} - H(\key\shieldB M)_{\rhomeasured}]
    \;,\]
    where
    $H(S)_\rho \coloneqq - \tr [\rho _S \log\rho _S]$
    is the quantum entropy of $\rho _S$ on system \emph{S}.
    However, because the key systems are in a mixture of Bell states,
    tracing out $\keyA$ will leave $\keyB$ in product with $\shieldB M$.
    Furthermore, the mixture of Bell states is uniform and thus
    $H(\keyB) _{\rhomeasured} = H(\key) _{\rhomeasured}$,
    therefore: 
    \begin{align}
    \label{eq:datahiding-belldiagonal:tmp}
    H(\keyB\shieldB M)_{\rhomeasured}
    = H(\key) _{\rhomeasured} + H(\shieldB M)_{\rhomeasured} \;.
    \end{align}
    We now use that
    \[H(X)_\alpha + H(Y)_\alpha - H(XY)_\alpha
    = D(\alpha _{XY}\parallel \alpha _X \otimes \alpha _Y)\]
    (where $\alpha$ is a state on $XY$) and conclude:
    \begin{align}
    E_D^\to(\rho)
    &\geq
    [H(\key) _{\rhomeasured} + H(\shieldB M)_{\rhomeasured}
    - H(\key\shieldB M)_{\rhomeasured}]
    \nonumber
    \\
    &= 
    D\quantity(
    {\textstyle \frac{1}{|\keyB|} \sum _\kindex}
    \phi _\zkindex {\otimes} \measurement(\sigma _\kindex)
    \middle\| 
    {\textstyle \frac{1}{|\keyB|} \sum _\kindex}
    \phi _\zkindex {\otimes} \measurement(\sigma))
    \nonumber
    \\
    &= 
    \frac{1}{|\keyB|} \sum \nolimits_\kindex
    D(\measurement(\sigma _\kindex) \parallel \measurement(\sigma))
    \end{align}
    where $\sigma = \frac{1}{|\keyB|} \sum_\kindex \sigma_\kindex$.
    Taking the supremum over $\measurement$ proves the claim.
\end{proof}

\begin{corollary}%
    \label{lemma:datahiding-correlated}
    For any key correlated state $\rho$, it holds:
    \begin{align*}
    E_D^\to(\rho)
    &\geq \sup  _{\measurement \in LO_A} \frac{1}{2^m} \sum _\kindex
    D(\measurement(\rho _\kindex)\parallel \measurement(\hat\rho))
    \end{align*}
    where $\rho _\kindex\coloneqq\phaseflip^\kindex_\keyB (\rho)$.
\end{corollary}
\begin{proof}
    We can use \Cref{lemma:datahiding-belldiagonal}
    after using \Cref{lemma:reversible-rho}:
    \begin{align}
    E_D^\to(\rho)
    &= E_D^\to(\reversiblemap(\rho))
    = E_D^\to\quantity({\textstyle \sum_\kindex}
    \tfrac{1}{|\keyB|} \phi_\zkindex 
    \otimes \rho _\kindex)
    \nonumber
    \\&\geq \sup  _{\measurement \in LO_A}
    \frac{1}{|\keyB|} \sum _\kindex
    D(\measurement(\rho _\kindex)\| \measurement(\hat\rho)).
    \end{align}
    where we used that
    $\hat\rho = \frac{1}{|\keyB|} \sum_\kindex \rho _\kindex$.
\end{proof}

These bounds generalize to the two-way distillable
entanglement because we can always apply \Cref{lemma:datahiding-correlated}
after a two-way preprocessing as shown below.

\begin{corollary}%
    \label{lemma:datahiding-correlated-twoway}
    For any key correlated state $\rho$, it holds:
    \begin{align*}
    E_D(\rho) 
    &\geq \sup  _{\measurement_{\measurements\ab}} 
    \frac{1}{2^m} \sum _\kindex
    D(\measurement(\rho _\kindex)
    \parallel \measurement(\hat\rho))
    \end{align*}
    where $LOCC _{\measurements\ab}$ are all two-way LOCC protocol
    that ends with a measurement at Alice's.
\end{corollary}
\begin{proof}
    By \Cref{lemma:reversible-rho}:
    \begin{align*}
    E_D(\rho)
    &=    E_D(\reversiblemap(\rho))
    \\&=    E_D({\textstyle \sum_\kindex}
    \tfrac{1}{2^m} \phi_\zkindex \otimes \rho _\kindex)
    \\&\geq \sup_{\Lambda} 
    E_D^\to({\textstyle \sum_\kindex}
    \tfrac{1}{2^m} \phi_\zkindex \otimes \Lambda(\rho _\kindex))
    \refstepcounter{equation}
    \tag\theequation
    \label{eq:implicit-step}
    \intertext{where $\Lambda$ are arbitrary LOCC protocols.
        We then use \Cref{lemma:datahiding-belldiagonal}:}
    E_D(\rho)
    &\geq \sup_{\Lambda} \sup  _{\measurement'}
    \frac{1}{2^m} \sum _\kindex
    D(\measurement'\circ\Lambda(\rho _\kindex)\parallel \measurement'\circ\Lambda(\hat\rho))
    \intertext
    {however without loss of generality we can write
        any measurement in $\locc_\ab$ as $\measurement'\circ\Lambda$,
        where $\measurement'$ is a measurement at Alice's
        and $\Lambda$ an $\locc_\ab$ protocol, therefore}
    E_D(\rho)
    &=    \sup  _{\measurement_{\measurements\ab}} 
    \frac{1}{2^m} \sum _\kindex
    D(\measurement(\rho _\kindex)\parallel \measurement(\hat\rho))
    \;.\qedhere
    \end{align*}
\end{proof}


We can finally give a rigorous proof of the main theorem.

\begin{theorem}[Main text \Cref{lemma:datahiding-oneway}]%
    \label{lemma:datahiding-oneway-full}
    For any key correlated state $\rho$, it holds:
    \begin{align}
    \label{eq:DA-singlecopy}
    E_D^\to(\rho) 
    &\geq
    D_A(\rho\parallel\hat\rho)
    \coloneqq
    \sup  _{\measurement \in LO_A}
    D(\measurement(\rho)\parallel \measurement(\hat\rho))
    \\
    \label{eq:DA-regularized}
    E_D^\to(\rho) 
    &\geq D_A^{\infty}(\rho\parallel\hat\rho)
    \coloneqq
    \lim _{n\to\infty} \frac{1}{n}
    D_A(\rho^{\otimes n}\parallel\hat\rho^{\otimes n})
    \end{align}
    If $\hat\rho$ is also separable then:
    \begin{equation}
    \label{eq:DA-equality}
    E_D^\to(\rho) =    D_A^\infty(\rho\parallel \hat\rho) \;.
    \end{equation}
\end{theorem}
\begin{proof}
    It is straightforward to check that
    $\phaseflip^\kindex_\keyB(\hat\rho) = \hat\rho$.
    Let $\measurement$ be any measurement $\measurement$ at Alice.
    Because the measurement is local at Alice,
    it commutes with the unitary $\phaseflip^\kindex$ at Bob,
    thus we have:
    \begin{align*}
    D(\measurement (\rho _\kindex)
    \|\measurement (\hat\rho))
    &=  D(\measurement \otimes \mathcal Z_{\keyB}^\kindex (\rho)
    \|\measurement \otimes \mathcal Z_{\keyB}^\kindex (\hat\rho))
    \\&=  D(\measurement (\rho) \|\measurement (\hat\rho)) \;.
    \end{align*}
    where that the relative entropy is unitary invariant.
    We can now rewrite \Cref{lemma:datahiding-correlated} as:
    \begin{align*}
    E_D^\to(\rho)
    &\geq \sup _{\measurement\in LO_A}
    \tfrac 1{2^m} {\textstyle \sum} 
    D(\measurement (\rho _\kindex)
    \|\measurement (\hat\rho))
    \\&=    \sup _{\measurement\in LO_A}
    \tfrac 1{2^m} {\textstyle \sum} 
    D(\measurement (\rho)
    \|\measurement (\hat\rho))
    \\&=    \sup _{\measurement\in LO_A}
    D(\measurement (\rho)
    \|\measurement (\hat\rho))
    \end{align*}
    proving \Cref{eq:DA-singlecopy}.
    For \Cref{eq:DA-regularized} we have:
    \[E_D^\to(\rho) = \frac{1}{n} E_D^\to(\rho^{\otimes n})
    \geq \frac{1}{n} D_A(\rho^{\otimes n}\parallel \hat\rho^{\otimes n})
    \qquad \forall\ n\]
    because the distillable entanglement is already regularized and
    $\rho^{\otimes n}$ is still a key correlated state.
    By Fekete's Lemma~\cite{fekete},
    $\frac{1}{n} D_A(\rho^{\otimes n}\|\hat\rho^{\otimes n}) $
    converges because $D_A$ is super-additive;
    taking the limit $n\to\infty$ proves \Cref{eq:DA-regularized}.
    Equality in \Cref{eq:DA-equality} follows because 
    if $\hat\rho$ is separable,
    then we get the opposite inequality from \Cref{lemma:regularized-1Er}:
    \begin{equation*}
    D_A^\infty(\rho\|\hat\rho) 
    \geq E_D^\to(\rho)
    \geq D_A^\infty(\rho\|\hat\rho) 
    \;.
    \qedhere
    \end{equation*}
\end{proof}

\begin{module}
\bigskip
We now show with an example that using the reversible map is often necessary.
We can provide an example of a {\bell} {\privatebit}
for which the bound of \Cref{lemma:datahiding-belldiagonal} is strictly suboptimal,
while the bound of \Cref{lemma:datahiding-correlated} achieves equality.

\begin{example}
    Let $|\keyA| = |\keyB| = |\shieldA|^2 = |\shieldB|^2 = 2^{2m}$,
    namely the key systems have now $2m$ qubits each
    while the shield have only $m$ qubits each.
    For the following example we need to use
    the whole Bell basis.
    We define the {\privatebit}:
    \begin{align*}
    \gamma^{2m} 
        &= \sum_{ij} \frac{1}{2^{2m}} 
           \phi _{0i} \otimes \phi _{0j} \otimes \phi_{ij}
    \\
    \hat \gamma ^{2m} 
        &= \hat\Phi ^m \otimes \hat\Phi^m \otimes \tau^m
    \end{align*}
    where $\tau^m = \frac{\one}{2^{2m}} = \sum _{ij} \frac{1}{2^{2m}} \phi _{ij} $
    is the maximally mixed state of the shield.
    By the optimality result of~\cite[Equation 8]{accessible-info},
    the bound of \Cref{lemma:datahiding-belldiagonal} computes to:
    \[\sup  _{\measurement_A} 
    \frac{1}{2^{2m}} \sum _{ij}
    D(\measurement(\phi _{ij})\parallel \measurement(\tau^m))
    = m
    \;.\]
    which is achieved measuring the computational basis.
    However, this state is distillable into $2m$ {\eprbit}s
    with just a sequence of unitaries.
    When changing from the computational basis 
    to the conjugate basis on both sides (bilaterally),
    $\phi_{ij}$ converts to $\phi_{ji}$
    and recall that $\bnot \phi_{0j}\otimes \phi_{kl} \bnot^\dagger = \phi_{0,j-l}\otimes \phi_{kl}$ (\Cref{lemma:bnot}).
    The unitary that distills $2m$ {\eprbit}s is then obtained by the following sequence:
    \begin{itemize}
        \item Applying $\bnot$ to the second and third Bell states results in:
         \[\gamma^{2m} \to \sum_{ij} \frac{1}{2^{2m}} \phi _{0i} \otimes \phi _{00} \otimes \phi _{ij}\]
        \item Applying the bilateral change of basis on the third Bell state results in:
        \[\qquad \sum_{ij} \frac{1}{2^{2m}} \phi _{0i} \otimes \phi _{00} \otimes \phi _{ij}
        \to\sum_{ij} \frac{1}{2^{2m}} \phi _{0i} \otimes \phi _{00} \otimes \phi _{ji}\]
        \item Applying $\bnot$ to the first and third Bell states results in:
        \[\qquad \sum_{ij} \frac{1}{2^{2m}} \phi _{0i} \otimes \phi _{00} \otimes \phi _{ij}\to \sum_{ij} \frac{1}{2^{2m}} \phi _{00} \otimes \phi _{00} \otimes \phi _{ij}\]
    \end{itemize}
    Namely, there exist a local unitary achieving the transformation
    \begin{align*}
    \gamma^{2m} \to \phi_{00} \otimes \phi_{00} \otimes \tau^m
    \end{align*}
    thus proving that $E_D(\gamma^{2m}) = 2m$.

    We now compute the bound of \Cref{lemma:datahiding-correlated}.
    We now have the {\privatebit}:
    \begin{align*}
    \mathcal{E}(\gamma^{2m}) 
        &= \sum_{kl} \frac{1}{2^{2m}} 
        \phi _{0k} \otimes \phi _{0l} \otimes \gamma^{2m}_{kl}
    \\
    \mathcal{E}(\hat\gamma ^{2m})
        &= \hat\Phi ^m \otimes \hat\Phi^m \otimes \hat\gamma^{2m}
    \intertext{where:} 
    \gamma^{2m}_{kl} &= \sum_{ij} \frac{1}{2^{2m}} \phi _{0,i+k} \otimes \phi _{0,j+l} \otimes \phi_{ij}
    \\
    \hat \gamma ^{2m}_{kl} &= \hat\Phi ^m \otimes \hat\Phi^m \otimes \tau^m
    \;.
    \end{align*}
    With the the same distillation procedure as above we can achieve the transformation
    \[\gamma^{2m}_{kl} \rightarrow \phi_{0k} \otimes \phi_{0l} \otimes \tau^m \;.\]
    This allows to compute the bound of \Cref{lemma:datahiding-correlated} as:
    \allowdisplaybreaks%
    \begin{align*}
    \;&\!\!
    \sup  _{\measurement_A} 
    \frac{1}{2^{2m}} \sum _{kl} 
    D(\measurement(\gamma^{2m}_{kl})\parallel \measurement(\hat \gamma^{2m}))
    \\&=  \sup  _{\measurement_A} 
    \frac{1}{2^{2m}} \sum _{kl} 
    D(\measurement(\phi_{0k} \otime \phi_{0l} \otime \tau^m)
    \|\measurement(\hat\Phi^m \otime \hat\Phi^m \otime \tau^m))
    \\&=  \sup  _{\measurement_A} 
    \frac{1}{2^{2m}} \sum _{kl} 
    D(\measurement(\phi_{0k} \otime \phi_{0l} )
    \|\measurement(\hat\Phi^m \otime \hat\Phi^m ))
    = 2m
    \end{align*}
    achieved measuring the conjugate basis. 
    Notice that Alice's distilling unitary
    can be done by the measurement,
    while Bob unitary can be done 
    because of unitary invariance of the relative entropy.
    This bound now optimal and performs strictly better than \Cref{lemma:datahiding-belldiagonal}.
\end{example}

\end{module}


\section{{Key repeater}}

\begin{module}
    \repeatersystems%
    \renewcommand{\charlie}{{\bm{C}}}
    \renewcommand{\acb}{{A\mathrlap{:\charlie:B}}}
    In this Appendix we motivate the details leading to \Cref{lemma:rded-oneway}.
    The {\rkey}s from~\cite{limits}
    are defined in the same fashion as the rates from \Cref{sec:distillation-rates}.
    Recall that we now have three parties:
    Alice and Charlie ($A$ and $\charlieA$) share $\rhoA$
    and  Charlie and Bob ($\charlieB$ and $B$) share $\rhoB$.
    The {\rkey} is defined as a \textrate{{\privatebit}s},
    under tripartite LOCC maps, where Charlie ($\charlie=\charlieA\charlieB$)
    is traced out at the end of the protocol, 
    see also \Cref{figure:repeaterappendix}.
    Formally,
    \[R_D(\rhoA,\rhoB) {\coloneqq}
    \distillrate{(\rhoA\otime\rhoB)}{\gamma},\]
    which again equals
    \[R_D(\rhoA,\rhoB) {=}
    \distillrate{(\rhoA\otime\rhoB)}{\!\bgamma}.\]
    
    \begin{figure}
        \newcommand{\circuitC}{.85em}
        \newcommand{\target}{\gamma^{n\achievedrate}}
        
        \begin{module}
            \repeatersystems%
            \providecommand{\puremultigate}[2]{*+<1em,.9em>{\hphantom{#2}} \POS[0,0]="i",[0,0].[#1,0]="e",!C *{#2},"e"+UR;"e"+UL **\dir{-};"e"+DL **\dir{-};"e"+DR **\dir{-};"e"+UR **\dir{-},"i"} 
            \newcommand{\dw}{\ar@{.} [0,-1]&\ar@{.} [0,-1]} 
            \newcommand{\phambda}{{\scriptstyle\phantom\Lambda}}
            \providecommand{\circuitC}{.7em}
            \providecommand{\circuitR}{.3em}
            \providecommand{\circuitRfix}{.7em}
            \providecommand{\target}{\Phi^{ne}}

            \[
            \phantom{{\rhoB}^{\otimes n}}\hspace{.5em}
            \begin{aligned}
            \Qcircuit  @C=\circuitC @R=\circuitR @!R=\circuitRfix {
                &\ustick{\scriptstyle {Alice}\;\;}
                &\qw& \multigate{+1}\empty &&&&\qw[-4]&\dw& \multigate{+1}\empty &\qw&\qw
                \\
                &&&\pureghost\empty&\controlo\cw&&&&&\controlo&\pureghost\empty\cw
                \\\lstick{{\rhoA}^{\otimes n}}\ar @{-} [-2,1] \ar @{-} [2,1]&&&&&&&&&\\
                &&&&&\controlo\ar@{=}[-2,-1]&\puremultigate{4}\empty\cw&\controlo\cw&
                \\
                &&&&&\qw[-4]&\ghost\empty&\qw&\dw&\qw&\multimeasureD{2}{\scriptstyle \trc}
                \\&{\scriptstyle Charlie\quad} &&&&&\pureghost\empty& \scriptstyle \mathrlap\distillmaps\\
                &&&&&\qw[-4]&\ghost\empty&\qw&\dw&\qw&\ghost{\scriptstyle \trc}
                \\
                &&&&&\controlo\ar@{=}[2,-1]&\pureghost        \empty\cw&\controlo\cw&
                \\\lstick{{\rhoB}^{\otimes n}} \ar @{-} [-2,1] \ar @{-} [2,1]&&&&&&&&&\\
                &&&\pureghost\empty&\controlo\cw&&&&&\controlo&\pureghost\empty\cw
                \\
                &\dstick{\scriptstyle Bob \;}
                &\qw& \multigate{-1}\empty &&&&\qw[-4]&\dw& \multigate{-1}\empty &\qw&\qw
                \gategroup{1}{4}{11}{11}{.5em}{--}
                \save "2,9"."3,8"  \ar@{:}"4,8"  \restore 
                \save "4,9"."3,10" \ar@{:}"2,10" \restore 
                \save "10,9"."9,8" \ar@{:}"8,8"  \restore 
                \save "8,9"."9,10" \ar@{:}"10,10"\restore 
            }
            \end{aligned}
            \;\;\approx _\eps
            \target\hspace{.25em}
            \begin{aligned}
            \Qcircuit  @C=\circuitC @R=\circuitR @!R=\circuitRfix {
                &\ustick{\scriptstyle Alice}&\qw
                \\\\\\\\\\
                \ar @{-} [-5,1] \ar @{-} [5,1]
                \\\\\\\\\\ \push{\rule{0pt}{1.08em}}
                &\dstick{\scriptstyle Bob}  &\qw
            }
            \end{aligned}
            \]
        \end{module}
        
        \caption{\label{figure:repeaterappendix}%
            Distillation of key in a single node repeater.\hfill}
        
    \end{figure}
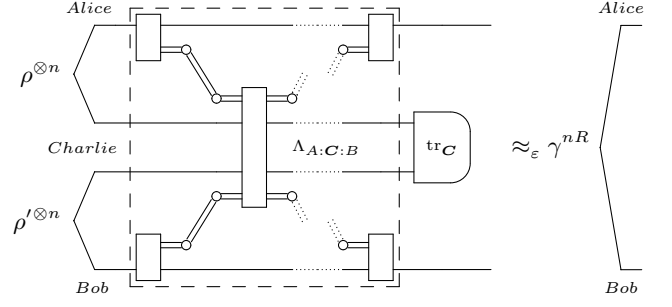
    
    Let us now indicate with $\locc _{\charlie\to\ab}$
    the tripartite \locc\ protocols that have the communication
    between Charlie and Alice/Bob restricted to be one-way from Charlie.
    The corresponding one-way {\rkey} is then equal to:
    {\renewcommand{\distillmaps}{\Lambda_{\charlie\mathrlap{\to\ab}}}
        \[R_D^\to(\rhoA,\rhoB) {=}
        \distillrate{(\rhoA\otime\rhoB)}{\gamma}.\]
    }
    The following reformulation is useful for the proof of the next theorem.    
    \begin{lemma}%
        \label{lemma:repeater-oneway-measurement}
        \renewcommand{\distillmap}{\Lambda{\circ}\measurement}
        \renewcommand{\distillmaps}{\substack{\measurement_{\mathrlap\charlie}%
                \\\Lambda_{\mathrlap\ab}}}
        \[R_D^\to(\rhoA,\rhoB) {=}
        \distillrate{(\rhoA\otime\rhoB)}{\gamma}.\]
    \end{lemma}
    \begin{proof}
        An arbitrary one-way protocol $\Lambda'\in\locc_{\charlie\to\ab}$ consists
        of an instrument on $\charlie$ followed by an $\locc_\ab$ protocol
        that is allowed to act on the classical part of the instrument outcome (the communication).
        Let $I$ be the instrument from $\charlie$ to $\charlie M$,
        where $M$ is now the classical register.
        Let $\Lambda\in\locc_{M:\ab}$ be the second part of the protocol.
        Then:
        \begin{align*}
        \trc \Lambda'
        &= \trc \circ (\id_\charlie \otimes \Lambda) \circ (I \otimes \id_{AB})
        \\&= \Lambda \circ (\trc I \otimes \id_{AB})
        \\&= \Lambda \circ (\measurement \otimes \id_{AB})
        \end{align*}
        where we used that tracing the quantum part of the instrument
        gives a measurement $\measurement = \trc I$.
        Since $M$ is classical we have $\locc_{M:\ab}\equiv\locc _{M\ab}$.
        Because every pair $(\measurement_\charlie,\Lambda_\ab)$
        also defines a map in
        $\locc_{\charlie\to\ab}$, we have equality in the claim.
    \end{proof}
    
    We can now formulate the upper bound used to derive \Cref{lemma:rded-oneway},
    this upper bound is a corollary of the following theorem.
    \begin{theorem}[{{\cite[Theorem~4]{limits}}}]%
        \label{lemma:repeatersep-oneway:limits}
        \[R_D^\to(\rhoA, \rhoB)
        \leq E_{R,\charlie\to AB}^{\infty,\accb} (\rhoA \otimes \rhoB)
        \;.\]
    \end{theorem}
    
    \begin{corollary}
        \label{lemma:repeatersep-oneway:appendix}
        For any pair of states $\rhoA$ and $\rhoB$
        and any separable state
        $\sigma$ in $\sep_\repcutA$ or $\sep_\repcutB$:
        \[R_D^\to(\rhoA, \rhoB)
        \leq D_{\charlie}^\infty(\rhoA \otimes \rhoB \parallel \sigma)
        \;.\]
    \end{corollary}

    \begin{proof}
        \allowdisplaybreaks%
        This is a consequence of \Cref{lemma:repeatersep-oneway:limits}.
        The direct way of proving the claim is to first 
        adapt the original proof to obtain:
        \begin{align*}
        R_D^\to(\rhoA , \rhoB)
        &\leq E_{R,\ctab}^{\infty,\repcutA}
        (\rhoA \otimes \rhoB )
        \\
        R_D^\to(\rhoA , \rhoB)
        &\leq E_{R,\ctab}^{\infty,\repcutB}
        (\rhoA \otimes \rhoB )
        \;,
        \end{align*}
        then restrict the optimization over tensor product separable states $\sigma^{\otimes n}$
        and finally remove the measurement on the receiver side $AB$
        as done in \Cref{lemma:regularized-1Er}.
        
        \providecommand{\untwister}{\mathcal{T}}
        However, we think it is instructive to see a direct proof.
        Without loss of generality let $\sigma\text{ in }\sep_\repcutA$.
        According to \Cref{lemma:repeater-oneway-measurement},
        let the measurement $\measurement$ and map $\Lambda$ be
        the optimal $R_D^\to$ distillation protocols for given $\eps$ and $n$.
        Then, for any $\sigma\in\sep_\repcutA$:
        \begin{align*}
        \frac{1}{n} D_{\charlie}&
        ((\rhoA \otimes \rhoB)^{\otimes n}
        \parallel \sigma ^{\otimes n})
        \\&\geq \frac{1}{n} D(\measurement
        ((\rhoA \otimes \rhoB)^{\otimes n})
        \parallel \measurement(\sigma ^{\otimes n}))
        \\&\geq \frac{1}{n} D(\Lambda\circ \measurement
        ((\rhoA \otimes \rhoB)^{\otimes n})
        \parallel \Lambda\circ \measurement(\sigma ^{\otimes n}))
        \\&\geq \frac{1}{n} D(\tilde\gamma^{n\achievedrate}\parallel \tilde\sigma_n)
        \end{align*}
        {where
            $\tilde\gamma^{n\achievedrate} \coloneqq
            \Lambda\circ \measurement ((\rhoA \otimes \rhoB)^{\otimes n})$
            will be a state $\eps$-close to a {\privatebit} $\gamma^{n\achievedrate}$,
            and $\tilde\sigma_n$ in still a separable state.
            The first inequality follows by definition of $D_\charlie$
            and the second by monotonicity of the relative entropy.
            
            Ideally, at this point we would use asymptotic continuity
            to change $\tilde\gamma^{n\achievedrate}$ into $\gamma^{n\achievedrate}$
            at the cost of some factor that goes to zero
            in the limits $n\to\infty$ and $\eps\to 0$.
            However, the dimensions of the shield systems of $\gamma^{n\achievedrate}$
            are in principle unbounded, so that we cannot argue directly
            that these factors go to zero.
            We need to remove the shield systems first,
            exploiting that, by definition, $\gamma^{n\achievedrate}$
            is a twisted version of a {\eprbit} $\Phi^{nR}$.
            This is the same argument used in~\cite[Theorem 9]{paradigms}.
            
            Let us denote by $\untwister$ the map that inverts the twisting unitary
            and traces the shield. Then, by monotonicity of the trace distance, we have:
            \[\untwister (\tilde\gamma^{n\achievedrate})
            \approx_\eps \untwister(\gamma^{n\achievedrate})
            = \Phi^{nR}\;.\]
            Furthermore, while
            $\untwister(\tilde\sigma_n)$ 
            might not be separable anymore,
            $\untwister$ will still map $\sep_\ab$ into a convex set $\untwister(\ab)$.
            Again we can apply $\untwister$ by monotonicity of the relative entropy, thus we find:}
        \begin{align*}
        \frac{1}{n} D&_{\charlie}((\rhoA\otimes\rhoB)^{\otimes n}\parallel\sigma ^{\otimes n})
        \\&\geq \frac{1}{n}
        D(\untwister(\tilde\gamma^{n\achievedrate})
        \parallel \untwister (\tilde\sigma_n))
        \\&\geq \frac{1}{n} 
        E_R^{\untwister(\ab)}
        (\untwister(\tilde\gamma^{n\achievedrate}))
        \\&\geq \frac{1}{n} 
        E_R^{\untwister(\ab)}
        (\Phi^{n\achievedrate}) 
        - \frac{1}{n}\quantity( n\achievedrate \eps
        - \eta(\eps) + \eta(1+\eps))
        \\&\geq \frac{1}{n} n\achievedrate - \eps\achievedrate
        + \frac{1}{n}\quantity(\eta(\eps) - \eta(1+\eps))
        \\&=    (1-\eps)\achievedrate
        + \frac{1}{n}\quantity(\eta(\eps) - \eta(1+\eps))
        \end{align*}
        {where we used the asymptotic continuity
            of the relative entropy with respect to $\untwister(\ab)$
            (\Cref{lemma:continuity},~\cite{relent-continuity,tight-continuity}),
            and that $E_R^{\untwister(\ab)}(\Phi^{n\achievedrate})
            \geq n\achievedrate$~\cite[Lemma 7]{paradigms}.
            Taking the limit $n\to\infty$ then leads to:}
        \[ D_\charlie^\infty(\rhoA \otimes \rhoB \parallel \sigma)
        \geq (1-\eps)\achievedrate\]
        and taking the limit $\eps\to 0$ concludes the proof.
    \end{proof}
    
    The original version uses states $\sigma$ separable in the quadri-partite cut $\accb$.
    However, it is immediate to see why the argument works
    also for $\repcutA$ and $\repcutB$\,:
    the only thing required in the proof is the separability of
    $\Lambda(\sigma)$ in $\ab$ for any distillation protocol $\Lambda$.

\end{module}


\section{{Key Swapper}}

\noindent
Before we begin to talk about key swapping protocols we need to introduce some definitions.

%
%
\begin{module}
    \providecommand{\loccD}{\locc}
    The class of {\bell} {\privatebit}s is not the only restriction
    one can make to the class of {\privatebit}s.
    The class of \emph{irreducible} {\privatebit}s
    was defined in~\cite{paradigms}:
    \begin{definition}
        [{Irreducible {\privatebit}s~\cite{paradigms}}]
        \hfill\\
        A {\privatebit} $\gamma ^m$ is called irreducible if $K_D( \gamma ^m) = m$.
    \end{definition}
    These are the {\privatebit}s that are actually interesting in the definition of distillable key
    because they are the outcomes of the optimal distillation protocols.
    However, the only feasible way to prove that a {\privatebit} is irreducible
    it to upper bound the distillable key via some other entanglement measure.
    For example one can use the relative entropy of entanglement,
    which in~\cite{paradigms} it was shown to give a further upper bound
    in terms of the relative entropy of entanglement of the {\keyattacked}:
    \begin{equation}
    \label{eq:irred-Er}
    K_D(\gamma^m) \leq E_R(\gamma) \leq m + E_R(\hat\gamma) \;.
    \end{equation}
    In light of this technique and considering that
    we need to require that $\hat \gamma$ is separable
    to argue that $E_D = E_{R,\loccD}^\infty$,
    it is sensible to introduce the following definition:
    \begin{definition}
        [{Strictly irreducible {\privatebit}s}]
        We say a {\privatebit} $\gamma$ is strictly irreducible
        if $\hat \gamma$ is separable.
        We will denote these states with $\gstrict$
        or $\strict \gamma m$.
    \end{definition}
    Of course, all strictly irreducible {\privatebit}s are irreducible.
    Indeed, these are all the {\privatebit}s for which we can prove $K_D(\strict \gamma m)=m$
    via \Cref{eq:irred-Er}. 
    A simple example are all {\privatebit}s for which $\sigma$
    is the maximally mixed state $\tau$;
    indeed, we find immediately that
    \begin{align*}
    \hat\gamma^m
    &= \frac1{2^m} \sum \nolimits_i \ketbra{ii}{ii}
    \otimes U_i^{\phantom\dagger} \tau U _i ^\dagger
    \\
    &= \frac1{2^m} \sum \nolimits_i \ketbra{ii}{ii}
    \otimes \tau U_i^{\phantom\dagger} U _i ^\dagger
    \\
    &= \frac1{2^m} \sum \nolimits_i \ketbra{ii}{ii}
    \otimes \tau
    \\
    &= \hat\Phi^m \otimes  \tau
    \end{align*}
    is always separable.
\end{module}

\bigskip

%
%
\begin{module}
    \repeatersystems%
    \renewcommand{\charlie}{{\bm{C}}}
    \renewcommand{\infty}{{\mathchar"231\vphantom0}} 
    \renewcommand{\delta}{{\smash{\mathchar"10E}\vphantom0}} 
    \newcommand{\rateA}{r}
    \newcommand{\rateB}{\rateA\smash'}
    \newcommand{\securityA}{\epsilon}
    \newcommand{\securityB}{\securityA\smash'}
    \newcommand{\distillerA}{\Gamma}
    \newcommand{\distillerB}{\distillerA\smash'}
    \renewcommand{\distillmaps}{\Lambda_{\charlie\rightarrow\mathrlap\ab}}
    We now show that the one-way distillable entanglement
    upper bounds the one-way {\rkey}
    of all protocols that first distill
    strictly irreducible {\privatebit}s with Charlie
    and then try to apply a general repeater protocol.
    First, we define a rate for such protocols.
    \begin{definition}
        [{One-way \skey}]%
        \label{def:key-swapping}
        \hfill\\
        For all bipartite $\rhoA$ and $\rhoB$, we define the one-way {\skey}
        achieved with one-way key swapping protocols as:
        \newcommand{\phamma}{{\vphantom{\distillerB_{C}}}}
        \renewcommand{\to}{{\smash\rightarrow\phamma}}
        \[ \swapkey (\rhoA,\rhoB) \coloneqq
        \lim _{\substack{\eps\to 0 \\ \securityA \to 0 \\ \securityB\! \to 0 }}
        \lim _{ n  \to\infty}
        \sup _{\substack{\phamma\Gamma _{\charlieA\rightarrow\mathrlap A}
                \\ \distillerB _{\charlieB\!\rightarrow\mathrlap B}
                \\ \phamma \Lambda_\ab,\mathrlap{\measurement_\charlie}}} 
        \quantity{ \achievedrate\colon \!\!
        \substack{ 
            \Lambda \circ \measurement
            (\strict \gamma {n\rateA} \otimes \strict{ \gamma}{n\rateB})
            \approx _\eps \gamma ^{n\achievedrate}
            \\
            \Gamma({\rhoA} ^{\otimes n}) \approx _{\securityA} \strict \gamma {n\rateA}
            \\
            \distillerB({\rhoB} ^{\otimes n}) \approx _{\securityB} \strict \gamma {n\rateB}
        }}.\]
    \end{definition}

    
    \noindent Then we can state the result:

    \begin{theorem}%
        \label{lemma:key-swapping}
        \[\swapkey (\rhoA,\rhoB) \leq E_D^\rightarrow (\rhoA \otimes \rhoB)\]
    \end{theorem}
    \begin{proof}
        \newcommand{\untwister}{\mathcal{T}}
        \newcommand{\bstrict}[2]{\strict{\bm{#1}}{#2}}
        \newcommand{\reversibleAB}{\bm{\reversiblemap}}
        \newcommand{\rhoAB}{\bm\rhoA}
        \newcommand{\rateAB}{\bm\rateA}
        \newcommand{\securityAB}{{\bm\securityA}}
        \newcommand{\distillerAB}{\bm\distillerA}
        \newcommand{\Mshield}{{\measurement}} 
        \allowdisplaybreaks%
        First notice that the tensor product of two strictly irreducible {\privatebit}s
        is still a strictly irreducible {\privatebit}, namely
        $\strict \gamma a \otimes \strict \gamma b  = \strict \gamma {a+b}$.
        
        Then, for the sake of the proof,
        let us introduce the following convenient bold shorthand notation:
        \begin{align*}
        \rateAB     &\coloneqq \rateA + \rateB
        &
        \rhoAB      &\coloneqq \rhoA \otimes \rhoB
        \\
        \bstrict \gamma {n \rateAB}
        &\coloneqq \strict \gamma {n \rateA} \otimes \strict \gamma {n \rateB}
        &
        \distillerAB&\coloneqq \Gamma \otimes \distillerB
        \end{align*}
        then with $\securityAB \coloneqq \securityA + \securityB$ we have
        \begin{align}
        \distillerAB(\rhoAB^{\otimes n}) \approx_\securityAB \bstrict \gamma {n\rateAB}
        \label{eq:lemma:swap-boldeps}
        \;.
        \end{align}
        We also define:
        \[\reversibleAB = \reversiblemap _{\keyA    \keyof    \charlieA}
        \otimes \id        _{\shieldA \shieldof \charlieA}
        \otimes \reversiblemap _{\keyB    \keyof    \charlieB}
        \otimes \id        _{\shieldB \shieldof \charlieB} \]
        where $\reversiblemap$ is the reversible map
        of \Cref{lemma:reversible-rho}. 
        
        Proof idea: just like in the proof of \Cref{lemma:datahiding-belldiagonal},
        we exploit the idea of using the distillable entanglement as an upper bound on
        the protocol that performs a measurement on the shield followed by hashing.
        However, we need to insert a step in the proof
        to substitute the approximate private state
        with exact private state, otherwise the proofs of
        \Cref{lemma:datahiding-correlated} and \Cref{lemma:datahiding-oneway}
        do not work. We will do this at the level of the coherent information
        using its asymptotic continuity.
        
        First, we lower bound the one-way distillable entanglement
        with the coherent information, just like in \Cref{lemma:datahiding-correlated}.
        Let $\Lambda$, $\measurement$ and $\distillerAB$
        be the optimal key swappping maps for given $\securityAB$, $\eps$, and $n$:
        \allowdisplaybreaks%
        \begin{align*}
        E_D^\rightarrow(\rhoAB)
        &=
        \frac{1}{n} E_D^\rightarrow(\rhoAB ^{\otimes n})
        \geq
        \frac{1}{n}
        E_D^\rightarrow(\reversibleAB \circ \distillerAB(\rhoAB ^{\otimes n}))
        \\&\geq
        \frac{1}{n}
        I({\keyof{\charlie}\rangle ABM})
        _{\measurement\circ\reversibleAB \circ \distillerAB (\rhoAB ^{\otimes n})}
        \end{align*}
        where $I(X\rangle Y) = H(Y) - H(XY)$ is the coherent information,
        $\keyof\charlie = \keyof C \keyof \charlieB$ and $AB=\key\shield$.
        Here we used that $E_D^\rightarrow$ entanglement is a one-way LOCC monotone,
        that $\reversibleAB \circ \distillerAB $ is one-way LOCC,
        and we used the hashing protocol after performing the measurement $\measurement$
        on the shield systems at Charlie.
        
        Here is where we want to change the approximate {\privatebit}s into exact {\privatebit}s
        as mentioned before.
        In the form of~\cite{tight-continuity},
        asymptotic continuity of the coherent information~\cite{condinfo1}
        says that:
        \[|I({X\rangle Y})_\varrho - I({X\rangle Y})_\varsigma|
        \leq 2\securityAB \log |X| - \eta(\securityAB) + \eta(1+\securityAB)
        \]
        for arbitrary bipartite states such that
        $\varrho _{XY} \approx_{\securityAB} \varsigma _{XY}$;
        here $\eta(x) = x\log x$.
        Since combining \Cref{eq:lemma:swap-boldeps}
        and the monotonicity of the trace distance gives
        \[\measurement\circ\reversibleAB (\distillerAB (\rhoAB^{\otimes n}))
        \approx_\securityAB
        \measurement\circ\reversibleAB (\bstrict \gamma {n \rateAB})\;,
        \]
        we can use the asymptotic continuity,
        where dimension factor is now $\log |\keyof\charlie| = {n\rateAB}$,
        and get:
        \begin{align*}
        {E_D^\rightarrow(\rhoAB)}
        &\geq
        \frac{1}{n}
        I({{\keyof{\charlie}\rangle ABM}})
        _{\measurement\circ\reversibleAB( \bstrict \gamma {n \rateAB})}
        \\&\quad
        - \frac{1}{n} \quantity
        ( n\rateAB \securityAB - \eta(\securityAB) + \eta(1+\securityAB))
        \\&=
        \frac{1}{n}
        I({{\keyof{\charlie}\rangle ABM}})
        _{\measurement\circ\reversibleAB( \bstrict \gamma {n \rateAB})}
        + O(\securityAB)
        \;.
        \end{align*}
        
        Now, as shown in
        \Cref{lemma:datahiding-correlated} and \Cref{lemma:datahiding-oneway},
        we can rewrite the conditional information as a relative entropy
        and then, by the unitary invariance of the relative entropy,
        correct the phase flip on the $AB$ side. This results in:
        \begin{align*}
        {E_D^\rightarrow(\rhoAB)}
        &\geq
        \frac{1}{n}
        D(       \measurement(\bstrict       \gamma  {n \rateAB})
        \parallel\measurement(\bstrict {\hat \gamma} {n \rateAB}))
        + O(\securityAB)
        \\&\geq
        \frac{1}{n}
        D(       \Lambda\circ\measurement(\bstrict      \gamma  {n \rateAB})
        \parallel\Lambda\circ\measurement(\bstrict {\hat\gamma} {n \rateAB}))
        + O(\securityAB)
        \end{align*}
        where the last inequality holds
        because of the monotonicity of the relative entropy.
        Notice that $\Lambda\circ\measurement(\bstrict {\hat\gamma} {n \rateAB})$
        is a separable state, because we distilled to strictly irreducible {\privatebit}s.
        
        At this point, just like in~\cite[Theorem 9]{paradigms}
        and \Cref{lemma:repeatersep-oneway:appendix},
        because $\Lambda\circ\measurement(\bstrict\gamma  {n \rateAB})
        \approx_\eps \gamma^{n\achievedrate}$,
        we can lower bound the relative entropy as follows:
        \begin{align*}
        {E_D^\rightarrow(\rhoAB)}
        &\geq \achievedrate
        - \frac{1}{n} \quantity
        (n\achievedrate \eps - \eta(\eps) + \eta(1+\eps))
        + O(\securityAB)
        \\&\geq \achievedrate
        + O(\eps)
        + O(\securityAB)
        \end{align*}
        Taking the limits $n\to\infty$, $\securityAB\to 0$ and $\eps\to 0$
        concludes the proof.
    \end{proof}
\end{module}


\section{{Single-copy Key Repeater}}

\begin{module}
    \repeatersystems%
    We consider now yet another variation of the repeater:
    the \emph{single copy} {\rkey} $R_D^{sc}$.
    Instead of letting Charlie act jointly on arbitrary many copies,
    we restrict him to act only on a single copy of the states.
    This should model, for example, memory-less repeater stations
    that perform their operations fast and do not allow for distillation.
    Alice and Bob then proceed to distill key as usual with the outcome of the
    single copy protocol with Charlie.
    The single copy {\rkey} has been defined in~\cite{limits} as follows.
    
    \begin{definition}[Single-copy $\bm{R_D}$~\cite{limits}]
        \begin{align*}
        R_D^{sc} (\rhoA,\rhoB)\coloneqq
        &\sup _{\Lambda_\acb} K_D(\trc\Lambda(\rhoA\otimes\rhoB))
        \;.
        \end{align*}
    \end{definition}
    
    It is possible to prove an upper bound
    in terms of a single copy relative entropy measure.
    We follow the proof of a similar upper bound that can be found in~\cite{limits}
    for the general {\rkey}, the main difference is that in this case
    there is no regularization, which would make the bound intractable for our purposes.
    We will later combine this bound with \Cref{lemma:datahiding-correlated}
    to express it in terms of the distillable entanglement.
    
    \providecommand{\trcab}{{\trc\cab}}
    \providecommand{\measurecab}{{\measurements\cab}}
    \begin{theorem}%
        \label{lemma:repeatersep-singlecopy}
        \newcommand{\distinguishmaps}{\measurecab}
        For all states $\rhoA$ and $\rhoB$ it holds:
        \begin{align*}
        R_D^{sc} (\rhoA,\rhoB) \leq E_{R,\distinguishmaps}^\repcutA(\rhoA\otimes\rhoB)
        \\
        R_D^{sc} (\rhoA,\rhoB) \leq E_{R,\distinguishmaps}^\repcutB(\rhoA\otimes\rhoB)
        \end{align*}
        where
        \begin{itemize}
            \item $E_{R,\distinguishmaps}^\repcutA$ and $E_{R,\distinguishmaps}^\repcutA$
            are relative entropies with $P=\sep_\repcutA $ and $P=\sep_\repcutB$
            respectively, and $\maps=\locc_\measurecab$;
            \item $\locc_\measurecab$ are all LOCC protocols of $\charlie{:}AB$
            that end with a measurement at Charlie.
        \end{itemize}
    \end{theorem}
    \begin{proof}
        \allowdisplaybreaks%
        Without loss of generality let $\sigma\in\sep_\repcutA$.
        Furthermore, for any map $\Lambda \in \locc_\acb$
        it holds that:
        \begin{align*}
        &\trc\Lambda\in\locc_\measurecab
        \\
        \tilde\sigma=
        &\trc\Lambda(\sigma) \in\sep_\ab
        \;.
        \end{align*}
        Therefore, for any such $\sigma$ and $\Lambda$ we have:
        \begin{align*}
        D_\measurecab&(\rhoA\otimes\rhoB\parallel \sigma ) 
        \\& = \sup_{\measurement _\measurecab}
        D(\measurement(\rhoA\otimes\rhoB)\parallel \measurement(\sigma )) 
        \\& \geq
        D(\trc\Lambda(\rhoA\otimes\rhoB)\parallel \trc\Lambda(\sigma)) 
        \\& =
        D(\trc\Lambda(\rhoA\otimes\rhoB)\parallel \tilde\sigma) 
        \\& \geq    
        E_R(\trc\Lambda(\rhoA\otimes\rhoB))
        \\& \geq 
        K_D (\trc\Lambda(\rhoA\otimes\rhoB))
        \end{align*}
        where we used that $E_R$ is a known upper bound on $K_D$~\cite{paradigms}.
        Taking the supremum over all $\Lambda$'s we find:
        \[  D_\measurecab(\rhoA\otimes\rhoB\parallel \sigma_\repcutA)
        \geq R_D^{sc} (\rhoA,\rhoB) \;.\]
        Taking the infimum over $\sigma$'s we end the proof.
    \end{proof}

    Recall now that \Cref{lemma:datahiding-correlated}
    generalizes to the two-way case in the following way:
    \begin{align*}
    E_D(\rho) 
    &\geq
    \sup _{\measurement_A \in \locc_\ab}
    \frac{1}{2^m} \sum _\kindex
    D(\measurement(\rho _\kindex)\parallel \measurement(\hat\rho))
    \\&=
    \sup _{\measurement_{\measurements\ab}}
    \frac{1}{2^m} \sum _\kindex
    D(\measurement(\rho _\kindex)\parallel \measurement(\hat\rho))
    \end{align*}
    
    As a direct application
    of \Cref{lemma:datahiding-correlated-twoway}
    to \Cref{lemma:repeatersep-singlecopy},
    we have now the following corollary.
    
    \begin{corollary}%
        \label{lemma:rded-singlecopy}
        Let $\rhoA$ 
        and $\rhoB$ 
        be any pair of key correlated states
        with at least one separable {\keyattacked}. Then:
        \[R_D^{sc}(\rhoA,\rhoB) \leq {|\keyB|}\cdot E_D(\rhoA \otimes \rhoB)\;.\]
    \end{corollary}
    \begin{proof}
        \begin{align*}
        R_D^{sc}(\rhoA,\rhoB)
        &\leq  D_\measurecab(\rhoA\otimes\rhoB\parallel \hat\rhoA \otimes \hat{\rhoB})
        \\&=     \sup _{\measurement_\measurecab} D(\measurement(\rhoA\otimes\rhoB)\parallel \measurement(\hat\rhoA \otimes \hat{\rhoB}))
        \\&\leq  \sup _{\measurement_\measurecab} \sum_{jk} D(\measurement(\rhoA_j\otimes\rhoB_k)
        \parallel \measurement(\hat\rhoA \otimes \hat\rhoB))
        \\&\leq  |\keyB| \cdot E_D(\rhoA \otimes \rhoB)
        \qedhere
        \end{align*}
    \end{proof}
    
\end{module}


\section{{Examples}}

\begin{module}
    \newcommand{\fourier}{{\mathbb{U}^\Gamma}} 
    \newcommand{\pptg}{\xi_\fourier}
    \newcommand{\pptinvg}{\xi_\Gamma}
    
    \newcommand*{\flip}{\mathbb{S}}
    \newcommand{\flipg}{\gamma_\flip}
    \begin{example}[{The Swap {\privatebit}s $\bm\flipg$~\cite{secure}}]\hfill\\
        This is a class of {\bell} {\privatebit}s with $m=1$,
        for each dimension $d=|\shieldA|=|\shieldB|$ it defines:
        \begin{equation}
        \label{def:flip-gamma}
        \flipg
        = \frac{1}{2} \quantity( 1 + \frac{1}{d} ) \phi_+ \otimes \rho _{s}
        + \frac{1}{2} \quantity( 1 - \frac{1}{d} ) \phi_- \otimes \rho _{a}
        \end{equation}
        for each dimension $d>1$,
        where $\rho _{s}$ and $\rho _{a}$
        are  the symmetric and anti-symmetric states
        in $\complex^d \otimes \complex^d$,
        the extreme Werner states~\cite{wernerstates}.
        %
        In {\privatebit} form, they are defined by:
        \begin{align*}
        \sigma &
        = \frac{\one}{d^2}
        &
        \twisting &= \one_2 \otimes (\proj 0 \otimes \one + \proj 1 \otime \flip)
        \end{align*}
        where $\flip = \sum _{i,j=0}^{d-1} \ketbra{ij}{ji}$ is the \emph{swap} operator.
        Notice that the swap is unitary and hermitian,
        thus $\flip^\dagger\flip = \flip\flip^\dagger = \flip^2=\one$;
        this gives the following block form:
        \[
        \flipg = \frac{1}{2} \frac{1}{d^2}
        \matrixquantity
        [\one & 0 & 0 & \flip \\
        0 & 0 & 0 & 0 \\
        0 & 0 & 0 & 0 \\
        \flip & 0 & 0 & \one]
        \;.\]
        This is a strictly irreducible {\privatebit} so $K_D(\flipg) = 1$.
        By the log-negativity upper bound
        on distillable entanglement~\cite{logneg}, we have
        \[E_D(\flipg)\leq E_N(\flipg) = \log_2 \quantity(1+\frac{1}{d} )
        \]
        which vanishes for large enough $d$.
    
        For the Swap {\privatebit}s
        we can immediately apply \Cref{lemma:rded-oneway}.
        This is the same example from the main text.
        \begin{corollary}%
            \label{lemma:flipg}
            \[ R_D^\to(\flipg,\flipg) \leq 2 \log_2 \quantity(1+\frac{1}{d} )
            \leq 2\frac{\log_2 e}{d} \;.\]
        \end{corollary}
        \noindent Aside from being the first example
        of an upper bound on {\rkey} for NPT states,
        \Cref{lemma:flipg} also improves on the
        {single copy {\rkey}} upper bound
        previously known~\cite{limits}:
        \[ R_D^{sc}(\flipg,\flipg)
        \leq  O\quantity(\frac{\log_2 d}{d}) 
        \;.\qedhere\]
    \end{example}
    
    \bigskip
    
    \begin{module}
        \providecommand{\fourier}{{\mathcal{U}}} 
        \providecommand{\flower}{{\mathbb{U}}} 
        
        In general, any unitary matrix in $d$ dimensions can be used to define a {\privatebit}~\cite{lowdimensional}.
        Here we focus only on the following special case:
        \begin{align}
        U &= \sum \nolimits_{ij} \frac{1}{\sqrt d} u_{ij} \ketbra{i}{j}
        \label{eq:U-fourier-flower}
        \intertext{such that $|u_{ij}| =1$, an example being
            the discrete Fourier transform.
            For each such $U$ we then define the following operators,
            to be used in the examples to follow:}
        \flower  &\coloneqq \sum \nolimits_{ij} u_{ij} \ketbra{ii}{jj}
        \nonumber
        \\
        \fourier &\coloneqq \sum \nolimits_{ij} u_{ij} \ketbra{ij}{ji}
        \nonumber
        \end{align}
        {where $(\cdot) ^\Gamma$ denotes the partial transpose. 
            Notice that $\fourier$ is a unitary and $\frac{\flower}{\sqrt{d}}$ is unitary
            in the maximally correlated subspace, namely}
        \begin{align*}
        \frac{\flower}{\sqrt d} \frac{\flower ^\dagger}{\sqrt d}
        = \frac{\flower ^\dagger}{\sqrt d} \frac{\flower}{\sqrt d} 
        &= \one_\corr
        \\
        \fourier ^\dagger \fourier = \fourier \fourier ^\dagger 
        &= \one\;.
        \end{align*}

        \begin{example}
            [{The Fourier {\privatebit}s $\bm{\gamma_\fourier}$~\cite{lowdimensional}}]
            The class of Fourier {\privatebit}s defines for $m=1$,
            for each $d=|\shieldA|=|\shieldB|$
            and for each $U$ as in \Cref{eq:U-fourier-flower}:
            \begin{align*}
            \sigma &= \frac{\one}{d^2}
            &
            \twisting &= \one_2 \otimes (\proj 0 \otimes \one + \proj 1 \otimes \fourier)
            \end{align*}
            or in block form:
            \begin{equation}
            \label{def:fourier-gamma}
            \gamma_\fourier = \frac{1}{2} \frac{1}{d^2}
            \matrixquantity
            [\one & 0 & 0 & \fourier \\
            0 & 0 & 0 & 0 \\
            0 & 0 & 0 & 0 \\
            \fourier^\dagger & 0 & 0 & \one] \;.
            \end{equation}
            Notice that in general these are not {\bell} {\privatebit}s
            because $\fourier$ is in general not hermitian.
            $U$ is usually taken to be the discrete Fourier transform, thus the name.
            
            This is also a strictly irreducible {\privatebit}, thus $K_D(\gamma_\fourier) = 1 $,
            and again we have an upper bound on distillable entanglement via the log-negativity:
            \[ E_D(\gamma_\fourier) \leq E_N(\gamma_\fourier)
            = \log_2 \quantity(1+ \frac{1}{\sqrt d} )
            \;.\]
            
            \begin{corollary}
                \[R_D^\to(\gamma_\fourier,\gamma_\fourier)
                \leq 2 \log_2 \quantity(1+ \frac{1}{\sqrt d} )
                \leq 2 \frac{\log_2 e}{\sqrt{d}} \;.\qedhere\]
            \end{corollary}
        \end{example}

        \begin{example}
            [{The Flower {\privatebit}s $\bm{\gamma_\flower}$~\cite{lowdimensional}}]
            \hfill\\
            Similarly, the class of Flower {\privatebit}s defines
            for $m=1$, for each $d=|\shieldA|=|\shieldB|$
            and for each unitary $U$:
            \begin{align*}
            \sigma & = \frac{\one_\corr}{d}
            &
            \twisting &= \one_{2} \otimes (\proj 0 \otimes \one + \proj 1 \otimes \flower)
            \end{align*}
            or in block form:
            \begin{equation}
            \label{def:flower-gamma}
            \gamma_\flower = \frac{1}{2} \frac{1}{d}
            \matrixquantity
            [\one_\corr & 0 & 0 & \frac\flower{\sqrt d} \\
            0 & 0 & 0 & 0 \\
            0 & 0 & 0 & 0 \\
            \frac{\flower^\dagger}{\sqrt d} & 0 & 0 & \one_\corr]\;.
            \end{equation}
            Again, these are not {\bell} {\privatebit}s in general.
            The first such an example was the \emph{flower state}~\cite{paradigms},
            which is obtained when $U$ is tensor products of the Hadamard transform. 
            
            They are still strictly irreducible {\privatebit}s, thus $K(\gamma_\flower) = 1 $.
            However, for the same reason that makes the log-negativity of the Fourier {\privatebit}s small,
            the log-negativity of the Flower {\privatebit}s becomes large
            and thus it cannot be used to find a meaningful bound
            on the distillable entanglement:
            \[E_N(\gamma_\flower) = \log_2 \quantity(1+ {\sqrt d} ) \;.\]
            Indeed, this can be far from the relative entropy of entanglement.
            Just like for the flower state,
            $E_R(\gamma_\flower)=1$
            because the relative entropy of entanglement is non lockable
            and $\gamma_\flower$ becomes separable
            after measuring either key system
            in the computational basis~\cite{flowerstate}.
            On the other hand, we can actually compute
            the distillable entanglement explicitly
            via the hashing bound $H(B) - H(AB)$,
            because $\gamma_\flower$ has support only
            on the maximally correlated subspace of $\key\shield$~\cite{e-maxcorr}.
            Since $\frac{\fourier}{\sqrt d}$ is a unitary
            in the maximally correlated subspace,
            it can be diagonalized in this subspace
            with all diagonal elements of unit module.
            In short, we find
            \[H(\key\shield)_{\gamma_\flower} = \log_2 d\]
            while the marginals are maximally mixed so
            \[H(\keyB\shieldB)_{\gamma_\flower} = 1 + \log_2 d\;.\]
            Therefore, for the class of Flower {\privatebit}s:
            \begin{equation}
            \label{eq:E-flower-gamma}
            E_D(\gamma_\flower) 
            = R(\gamma_\flower,\gamma_\flower) 
            = K(\gamma_\flower) 
            = E_R(\gamma_\flower) = 1\;.
            \end{equation}
            While they might not seem interesting, we will need these states for the PPT invariant examples.
        \end{example}

    \end{module}
    
    \begin{module}
        \providecommand{\fourier}{{\mathcal{U}}} 
        \providecommand{\flower}{{\mathbb{U}}} 
        
        \providecommand{\pptg}{\xi_\fourier}
        
        \begin{example}[{The PPT (noisy) {\privatebit}s $\bm\pptg$~\cite{limits}}]%
            \label{def:pptinvg}
            \newcommand{\phrt}{\vphantom{\sqrt{d}}}
            These are not exact {\privatebit}s,
            they are approximate {\privatebit}s
            that can be made arbitrarily close to the Fourier {\privatebit}s
            while still being PPT\@.
            The class of PPT {\privatebit}s defines
            (for $m=1$, for each $d=|\shieldA|=|\shieldB|$
            and for each  $U$ as in \Cref{eq:U-fourier-flower}):
            \begin{equation}
            \label{def:ppt-gamma}
            \pptg = \frac{1}{1 + \frac{1}{\sqrt d}}
            \quantity( \gamma_\fourier + \frac{1}{\sqrt d} 
            X_{\keyA} \hat\gamma_\flower X_{\keyA} )
            \end{equation}
            where the function of the local bit flip
            is to move the {\keyattacked} $\hat\gamma_\flower$ in the orthogonal subspace.
            Namely, in block form:
            \[
            \pptg = \frac{1}{2}  \frac{1}{1 + \frac{1}{\sqrt d}}
            \matrixquantity
            [\frac{\one}{d^2} & 0 & 0 & \frac{\fourier}{d^2} \\
            0 & \frac{1}{\sqrt d} \frac{\one_\corr}{d\phrt} & 0 & 0 \\
            0 & 0 & \frac{1}{\sqrt d} \frac{\one_\corr}{d\phrt} & 0 \\
            \frac{\fourier^\dagger}{d^2} & 0 & 0 & \frac{\one}{d^2}]
            \;.\]
            One can check that this noise is just enough to make them PPT,
            and that remarkably, the amount of noise needed in the mixture goes  to zero for large $d$.
            The PPT {\privatebit}s are engineered to become close to the set of separable states after partial transposition.
            Indeed, since $\one$ and $\one_\corr$ are PPT invariant, we find
            \[
            \pptg ^\Gamma = \frac{1}{2}  \frac{1}{1 + \frac{1}{\sqrt d}}
            \matrixquantity
            [\frac{\one}{d^2} & 0 & 0 & 0 \\
            0 & \frac{1}{\sqrt d} \frac{\one_\corr}{d\phrt} &
            \frac{1}{\sqrt d} \frac{\phantom d \flower}{d \sqrt d} & 0 \\
            0 & \frac{1}{\sqrt d} \frac{\phantom d \flower^{\mathrlap\dagger}}{d \sqrt d} &
            \frac{1}{\sqrt d} \frac{\one_\corr}{d\phrt} & 0 \\
            0 & 0 & 0 & \frac{\one}{d^2}]
            \;.\]
            and thus:
            \[\pptg^\Gamma =  \frac{1}{1 + \frac{1}{\sqrt d}} \quantity
            (\hat\gamma_\fourier + \frac{1}{\sqrt d} X_{\keyA} \gamma_\flower X_{\keyA})
            \]
            which is suddenly mostly a separable {\keyattacked}
            with a vanishing mixture of a Flower {\privatebit}.
            
        Until now, the PPT {\privatebit}s were the only example in the literature
        for which the {\rkey} could be upper bound by a computable quantity.
        Because these {\privatebit}s are PPT, 
        the distillable entanglement is zero ($E_D(\pptg)=0$);
        however, $\pptg$ are not exact {\privatebit},
        so we cannot use \Cref{lemma:rded-oneway} directly anymore.
        Instead, we need to exploit
        the monotonicity of the {\rkey} under one-way LOCC operations
        and the fact that PPT {\privatebit}s are obtained
        by mixing the Fourier {\privatebit}s with noise
        via a one-way LOCC operation, thus we find:
        
        \begin{corollary}%
            \label{lemma:rd-ppt-gamma}
            \[R_D^\to(\pptg,\pptg)
            \leq 2\log_2 \quantity(1+ \frac{1}{\sqrt d} )
            \leq 2 \frac{\log_2 e}{\sqrt d }\;.\]
        \end{corollary}
        
        Because Charlie is traced out at the end
        of the key repeater distillation protocol,
        the {\rkey} is invariant
        under transposition of the input on Charlie's systems,
        thus giving the following upper bound on the {\rkey}~\cite{limits}:
        {\repeatersystems%
            \begin{equation}
            \label{bound:keypartiatranspose}
            R_D(\rhoA,\rhoB)
            \leq \min\{K_D({\rhoA}^\Gamma),K_D({\rhoB}^\Gamma)\}\;.
            \end{equation}}
        The previous upper bound on $R_D(\pptg,\pptg)$ was computed
        by estimating an upper bound on $K_D(\pptg^\Gamma)$:
        \[R_D(\pptg,\pptg) \leq 
        O\quantity(\frac{\log_2 d}{\sqrt{d}})
        \;.\]
        However this bound is not optimal;
        used properly, \Cref{bound:keypartiatranspose},\
        yields the following bound.
        
        \begin{corollary}
            \[R_D(\pptg,\pptg) \leq \frac{1}{\sqrt{d} + 1}\;.\]
        \end{corollary}
        \begin{proof}
            \renewcommand{\qedsymbol}{\ensuremath{\square}}
            By \Cref{bound:keypartiatranspose}
            and convexity of the relative entropy of entanglement we find:
            \begin{align*}
            R_D(\pptg,\pptg)
            &
            \leq K_D(\pptg^\Gamma)
            \\&
            \leq E_R(\pptg^\Gamma)
            \\&
            \leq \frac{1}{1 + \frac{1}{\sqrt d}} \quantity
            (E_R(\hat\gamma_\fourier)+\frac{1}{\sqrt d} E_R(\gamma_\flower))\;.
            \end{align*}
            However $\hat\gamma_\fourier$ is separable
            and, according to \Cref{eq:E-flower-gamma}, $E_R (\gamma_\flower)=1$.
            Therefore:
            \[R_D(\pptg,\pptg) \leq \frac{1}{\sqrt{d}+1}\;.\qedhere\]
        \end{proof}
        This bound is still better than \Cref{lemma:rd-ppt-gamma}
        and holds for two-way protocols.
        \end{example}

    \end{module}
    
    \begin{module}
        \providecommand{\fourier}{{\mathcal{U}}} 
        \providecommand{\flower}{{\mathbb{U}}} 
        
        \providecommand{\pptg}{\xi_\fourier}
        \providecommand{\pptinvg}{\xi_\Gamma}

        \begin{example}
            [{The PPT invariant (noisy) {\privatebit}s $\bm\pptinvg$~\cite{lowdimensional}}]
            \label{example:ppt-invariant}
            \newcommand{\phrt}{\vphantom{\sqrt{d}}}
            By substituting the {\keyattacked}
            with the Flower {\privatebit}
            in \Cref{def:ppt-gamma},
            the expression becomes PPT invariant.
            Namely, the class of PPT invariant {\privatebit}s defines
            \begin{equation}
            \pptinvg =  \frac{1}{1 + \frac{1}{\sqrt d}}
            \quantity( \gamma_\fourier + \frac{1}{\sqrt d} 
            X_{\keyA} \gamma_\flower X_{\keyA}^\dagger )
            = (\pptinvg) ^\Gamma\;.
            \end{equation}
            (where $X_{\keyA} $ is the bitflip on ${\keyA}$) with block form
            \[ \pptinvg =
            \frac{1}{2}  \frac{1}{1 + \frac{1}{\sqrt d}}
            \matrixquantity
            [\frac{\one}{d^2} & 0 & 0 & \frac{\fourier}{d^2} \\
            0 & \frac{1}{\sqrt d} \frac{\one_\corr}{d\phrt} &
            \frac{1}{\sqrt d} \frac{ \flower}{d \sqrt d} & 0 \\
            0 & \frac{1}{\sqrt d} \frac{ \flower^{\mathrlap\dagger}}{d \sqrt d} &
            \frac{1}{\sqrt d} \frac{\one_\corr}{d\phrt} & 0 \\
            \frac{\fourier^\dagger}{d^2} & 0 & 0 & \frac{\one}{d^2}]
            \]
            which is clearly PPT invariant.

        \newcommand{\sqrnorm}{\frac{1}{1 + \frac{1}{\sqrt{d}}}}

        The fact that these {\privatebit}s are PPT invariant makes
        \Cref{bound:keypartiatranspose} useless,
        but a bound can still be computed combining
        one-way LOCC monotonicity,  \Cref{lemma:rded-oneway}
        and the Rains bound:
        \begin{corollary}%
            \label{lemma:pptinvg}
            \[R_D^\to(\pptinvg,\pptinvg) \leq 2\frac{1 + \log_2 e }{1 + \sqrt d}\;.\]
        \end{corollary}
        
        \begin{proof}
            \renewcommand{\qedsymbol}{\ensuremath{\square}}
            We introduce the following new states:
            \begin{align*}
            \alpha  &= \sqrnorm \quantity(  \proj{0} \otimes \gamma_\fourier
            + \frac{1}{\sqrt d} \proj{1} \otimes \gamma_\flower  )
            \\
            \tilde\alpha
            &= \sqrnorm \quantity(  \proj{0} \otimes \pptg
            + \frac{1}{\sqrt d} \proj{1} \otimes \hat\gamma_\flower  )
            \;.\end{align*}
            Who holds the additional qubit is irrelevant,
            but by making it part of the shield it is possible to show
            that $\alpha$ is actually a strictly irreducible {\privatebit}.
            Indeed the {\keyattacked} is separable and we have 
            \begin{equation*}
            \alpha {\propto} \matrixquantity[ 
                {\proj{0}} \otime \one {+} {\proj{1}} \otime \frac{\one_\corr}{\sqrt d} 
                & 0 & 0 &
                {{\proj{0}} \otime \fourier {+} {\proj{1}} \otime \flower} \\
                0 & 0 & 0 & 0 \\
                0 & 0 & 0 & 0 \\
                \quantity( {\proj{0}} \otime \fourier {+} {\proj{1}} \otime \flower )^\dagger & 0 & 0 &
                {\proj{0}} \otime \one {+} {\proj{1}} \otime \frac{\one_\corr}{\sqrt d}
                .]
            \end{equation*}
            We can check that this is a {\privatebit} by checking that the above matrix satisfies the block form, as explained in \Cref{sec:privatestates}. This reduces to check the following equality 
            \begin{align*}
            \proj{0}& \otime \one {+} \proj{1} \otime \frac{\one_\corr}{\sqrt d} = 
            \\=&\sqrt{
                {\quantity(\proj{0}\otime\fourier + \proj{1}\otime\flower) ^\dagger}
                {\quantity(\proj{0}\otime\fourier + \proj{1}\otime \flower)}}
            \end{align*}
            which proves that $\alpha$ is a {\privatebit}.
            One can obtain $\pptinvg$ from $\alpha$ via LOCC: 
            use the additional qubit to bit flip $\gamma_\flower$ but not $\gamma_\fourier$
            and then trace the qubit.
            Furthermore, $\tilde\alpha$ is clearly PPT,
            since it is mixture of the PPT states $\pptg$ and $\hat\gamma_\flower$.
            We find that $\alpha$ is close to $\tilde\alpha$:
            \begin{align}
            D(\alpha\|\tilde\alpha)
            &
            = \sqrnorm \quantity (  D(\gamma_\fourier \| \pptg)
            + \frac{1}{\sqrt d} D(\gamma_\flower  \| \hat\gamma_\flower ))
            \nonumber
            \\
            &
            = \sqrnorm \quantity(  \log \quantity( 1 + \frac{1}{\sqrt d} )
            + \frac{1}{\sqrt d} E_R(\gamma_\flower))
            \nonumber
            \\&
            \leq \frac{\sqrt d}{1 + \sqrt d} \quantity
            ( \frac{1}{\sqrt d \ln 2} + \frac{1}{\sqrt d} )
            \nonumber
            \\&
            = \frac{1 + {\log_2 e} }{1 + \sqrt d} \;.
            \label{eq:fourier-gamma-bound}
            \end{align}
            Now,  we use one-way LOCC monotonicity of $R_D^\to$,
            \Cref{lemma:rded-oneway},
            $E_D^\to(\rho)\leq E_R^{PPT}(\rho)$~\cite{rains-set},
            the fact that $\tilde\alpha$ is PPT
            and \Cref{eq:fourier-gamma-bound},
            in this order,  to show the claim:
            \begin{align*}
            R_D^\to(\pptinvg,\pptinvg)
            &
            \leq R_D^\to(\alpha,\alpha)
            \\&
            \leq 2 E_D^\to(\alpha)
            \\&
            \leq 2 E_R^{PPT}(\alpha) 
            \\&
            \leq 2 D(\alpha\|\tilde\alpha)
            \\&
            \leq 2\frac{1 + {\log_2 e}}{1 + \sqrt d}
            \tag*\qedhere
            \end{align*}
        \end{proof}        
        \end{example}
    \end{module}
\end{module}

\end{document}